\documentclass{amsart}
\usepackage{amssymb}
\usepackage{amsthm}
\usepackage{amstext}
\usepackage{amsbsy}
\usepackage{amscd}
\usepackage{enumerate}
\usepackage[noadjust]{cite}
\usepackage{mathtools}
\usepackage[usenames,dvipsnames]{color}

\theoremstyle{plain}
\newtheorem{thm}{Theorem}[section]
\newtheorem{lemma}[thm]{Lemma}

\newtheorem{cor}[thm]{Corollary}

\newtheorem{prop}[thm]{Proposition}
\theoremstyle{definition}

\theoremstyle{remark}
\newtheorem{remark}[thm]{Remark}

\newcommand{\norm}[1]{\left\Vert#1\right\Vert}
\newcommand{\abs}[1]{\left\vert#1\right\vert}
\newcommand{\set}[1]{\left\{#1\right\}}
\newcommand{\R}{\mathbb{R}}
\newcommand{\C}{\mathbb{C}}

\newcommand{\Z}{\mathbb{Z}}

\def\Re{\mathop{\rm Re}\nolimits}

\DeclareMathOperator{\Ker}{Ker}

\DeclareMathOperator*{\supp}{supp}
\allowdisplaybreaks

\begin{document}

\title[Dynamics of Jacobi and block Jacobi Matrices]{Quantum Dynamics of Periodic and Limit-Periodic Jacobi and block Jacobi Matrices with Applications to Some Quantum Many Body Problems}

\author[D. Damanik]{David Damanik}
\email{damanik@rice.edu}
\address{Mathematics Dept. MS-136, Rice University, Houston, TX 77005}

\author[M. Lukic]{Milivoje Lukic}
\email{milivoje.lukic@rice.edu}
\address{Mathematics Dept. MS-136, Rice University, Houston, TX 77005}

\author[W. Yessen]{William Yessen}
\email{yessen@rice.edu}
\address{Mathematics Dept. MS-136, Rice University, Houston, TX 77005}

\thanks{D.D. was supported in part by NSF grant DMS--1067988, M.L. was supported in part by NSF grant DMS--1301582, and W.Y. was supported by NSF grant DMS--1304287}

\date{\today}

\begin{abstract}

We investigate quantum dynamics with the underlying Hamiltonian being a Jacobi or a block Jacobi matrix with the diagonal and the off-diagonal terms modulated by a periodic or a limit-periodic sequence. In particular, we investigate the transport exponents. In the periodic case we demonstrate ballistic transport, while in the limit-periodic case we discuss various phenomena such as quasi-ballistic transport and weak dynamical localization. We also present applications to some quantum many body problems. In particular, we establish for the anisotropic XY chain on $\mathbb{Z}$ with periodic parameters an explicit strictly positive lower bound for the Lieb-Robinson velocity.

\end{abstract}

\maketitle

\sloppy

\section{Introduction}\label{sec:intro}

Quantum spin systems have been investigated for a long time as models of interacting many-body quantum systems. Given that the state space of a single spin is a finite-dimensional Hilbert space, the algebra of local observables is particularly simple. This allows us to focus exclusively on the quantum-dynamical phenomena arising from interaction of the spins. One of the most widely studied dynamical properties of spin models is the so-called \textit{Lieb-Robinson bound} \cite{LR72}, which gives restrictions on the speed of propagation of local disturbances through the chain. The general form of the Lieb-Robinson bound is
\begin{align}\label{eq:lr-general}
\norm{[\tau_t(A), B]}\leq C\norm{A}\norm{B} e^{-\eta(d - v\abs{t})},
\end{align}
where $C$, $\eta$, and $v$ are positive constants, $d$ is the distance between the regions where observables $A$ and $B$ are supported, $\tau_t(A)$ is the time evolution of $A$ and $t$ is the time parameter. We will give more details later.

In some extensively studied one-dimensional models, the dynamics of the interacting spin chain is strongly related to the dynamics governed by a single particle Hamiltonian (via the Jordan-Wigner transformation and the Lieb-Schultz-Mattis ansatz \cite{LSM61}, which allows to map certain many-body models to free Fermion models). This allows one to relate the Lieb-Robinson bound to transport properties of the associated single particle Hamiltonian. This is the main focus of our investigation: \textit{we concentrate on the widely studied anisotropic $XY$ chain and investigate the Lieb-Robinson bound via the transport properties of the associated single particle Hamiltonian}. For example, in their recent work \cite{HSS12}, E. Hamza, R. Sims and G. Stolz demonstrated how dynamical localization in the associated one particle Hamiltonian implies an exponentially small (in the distance) bound on the propagation of local disturbances; namely, they proved that the bound \eqref{eq:lr-general} holds with $v = 0$. They then proved dynamical localization for the one particle Hamiltonian associated to the isotropic $XY$ chain with random interaction and/or magnetic field. One of our main results gives the converse for the Hamza-Sims-Stolz theorem. In particular, we show that in the case of periodic interaction and magnetic field, the Lieb-Robinson bound with $v = 0$ fails. In fact, a uniform positive lower bound on $v$ is given in terms of some transport properties of the associated one particle Hamiltonian (which turns out to be a periodic block Jacobi matrix exhibiting ballistic transport).

Let us briefly describe the general setup and state our main results formally.

For a finite set $S \subset\mathbb{Z}$, we consider the $2^{\lvert S\rvert}$ dimensional Hilbert space $\mathcal{H}^{(S)}=\bigotimes_{j\in S} \C^2$ and the algebra of observables $\mathcal{A}_S = \bigotimes_{j\in S} \mathcal{A}_j$ where $\mathcal{A}_j = M(2,\mathbb{C})$ is the algebra of $2\times 2$ complex matrices, viewed as acting on the $j$-th copy of $\mathbb{C}^2$.

For $S_1 \subset S_2$, $\mathcal{A}_{S_1}$ embeds into $\mathcal{A}_{S_2}$ by the map $A \mapsto A \otimes \mathbb{I}_{S_2 \setminus S_1}$. The standard convention is to treat this embedding as an identity, so that $\mathcal{A}_{S_1}$ is viewed as a subset of  $\mathcal{A}_{S_2}$. With that convention, the algebra of local observables on $\mathbb{Z}$ is defined as
\[
\mathcal{A}_{\mathbb{Z}} = \bigcup_{\substack{S \subset \mathbb{Z} \\ S\text{ is finite}}} \mathcal{A}_S
\]
and for $A \in \mathcal{A}_{\mathbb{Z}}$, $\supp A$ is the smallest set $S$ such that $A \in \mathcal{A}_S$.

We consider the 1-dimensional anisotropic $XY$ model, given for $\Lambda = [m,n] \cap \mathbb{Z}$ by the Hamiltonian $H^{(\Lambda)}$ acting on the $2^{\lvert \Lambda\rvert}$ dimensional Hilbert space $\mathcal{H}^{(\Lambda)}=\bigotimes_{j=m}^n \C^2$ as follows.
\begin{align}\label{eq:haml}
\begin{split}
  H^{(\Lambda)}= \sum_{j=m}^{n-1}\mu_j[(1+\gamma_j)\sigma_j^{(x)}\sigma_{j+1}^{(x)}+(1-\gamma_j)\sigma_j^{(y)}\sigma_{j+1}^{(y)}]+\sum_{j=m}^n\nu_j\sigma_j^{(z)},
\end{split}
\end{align}
where $\set{\mu_j}$ is the sequence of interaction couplings (notice that only nearest neighbor interactions are taken into account) and $\set{\nu_j}$ is the magnetic field in the direction transversal to the lattice; here $\set{\gamma_j}$ is the sequence of anisotropy coefficients.
The operators $\sigma_j^{(v)}$, $v\in\set{x,y,z}$, are defined by
\begin{align*}
\underbrace{\mathbb{I}\otimes\cdots\otimes\mathbb{I}}_{j-m\text{ times}}\otimes \sigma^{(v)}\otimes\underbrace{\mathbb{I}\otimes\cdots\otimes\mathbb{I}}_{n-j\text{ times}},
\end{align*}
where $\mathbb{I}$ is the $2\times 2$ identity operator and $\sigma^{(v)}$ are the Pauli matrices:
\begin{align*}
\sigma^{(x)}=\begin{pmatrix}
            0 & 1\\
            1 & 0\\
           \end{pmatrix},\hspace{2mm}
\sigma^{(y)}=\begin{pmatrix}
            0 & -i\\
            i & 0\\
           \end{pmatrix},\hspace{2mm}
\sigma^{(z)}=\begin{pmatrix}
            1 & 0\\
            0 & -1\\
           \end{pmatrix}.
\end{align*}

We assume that for all $j$, $\mu_j\neq 0$ (otherwise the chain breaks into a collection of noninteracting subchains). Notice also that with $\gamma_j\equiv 1$ we obtain the Ising model in a magnetic field, and with $\gamma_j\equiv 0$ we obtain the isotropic XY model.

We also need to define the interaction boundary of a set $S\subset\mathbb{Z}$, which in the case of the XY spin chain can be defined as
\[
\partial S = \{ n \in S : \{n-1, n+1 \} \not\subset S\}.
\]

\begin{thm}[Lieb--Robinson~\cite{LR72}, Nachtergaele--Sims~\cite{NS09}] \label{thm:LRintro}
For any $\Lambda = [m,n]\cap \mathbb{Z}$, any disjoint $S_1, S_2 \subset\Lambda$, any $A \in \mathcal{A}_{S_1}$ and $B \in \mathcal{A}_{S_2}$ and all $t\in \mathbb{R}$,
\begin{align}\label{eq:lr-spec}
\norm{[\tau^{(\Lambda)}_t(A), B]}\leq C\norm{A}\norm{B}e^{-\eta(d(S_1, S_2)-v\abs{t})}
\end{align}
with uniform constants $\eta, v > 0$ and a constant $C$ which can depend solely on the size of the interaction boundaries of $S_1, S_2$.
\end{thm}

\begin{remark}
If we restrict to observables with $\max S_1 < \min S_2$, the constant $C$ becomes uniform, since we can always replace $S_1, S_2$ by their convex hulls (intersected with $\mathbb{Z}$), after which the supports are still disjoint and have interaction boundaries of at most $2$ points each.
\end{remark}

The Lieb-Robinson bound \eqref{eq:lr-general} is by its very nature an upper bound and hence bounds transport from above. Lower transport bounds to complement such a statement are naturally of interest as well. One may capture such a complementary result by proving that certain upper bounds fail to hold or, put differently, by showing that in the Lieb-Robinson bound, which always holds for some constants due to Theorem~\ref{thm:LRintro}, the constants must be subject to some restrictions.

The most interesting constant to restrict is the velocity $v$. The Lieb-Robinson bound \eqref{eq:lr-general} with a positive velocity $v$ is a ballistic upper bound. Upper bounds for $v$ hold in general; compare for example \cite{Aizenman2012}. In some cases one expects genuine ballistic transport, and hence the natural result to establish in these cases is that there is some strictly positive lower bound for $v$ in any Lieb-Robinson bound that holds in the given context. Related to this, the absence of ballistic transport is related to the ability to make $v$ in the Lieb-Robinson bound arbitrarily small.

Periodic models are generally expected to lead to ballistic transport. This is suggested, in particular, by the (expectation of) ballistic transport for the associated (via the aforementioned Lieb-Schultz-Mattis ansatz) one particle Hamiltonian. Perhaps surprisingly, this statement has not yet been rigorously established in great generality (the case of the isotropic constant-coefficient XY model is discussed in Example~6.2.14B in \cite{BR97}). In particular, to the best of our knowledge, in the case of the anisotropic XY chain on $\mathbb{Z}$ considered here, this has not yet been addressed, even though the result is widely expected to hold among experts. This is one of our primary sources of motivation for writing this paper. Indeed we will address precisely this issue of proving a ballistic lower transport bound in the periodic case. Specifically, we prove

\begin{thm}\label{thm:main-1}
For the anisotropic XY chain on $\mathbb{Z}$ described above, with periodic 
sequences of $\mu_j \in \mathbb{R}\setminus \{0\}$, $\gamma_j \in \mathbb{R} \setminus\{\pm1\}$, $\nu_j \in \mathbb{R}$, there exists a $v_0>0$ such that if the Lieb-Robinson bound \eqref{eq:lr-spec} holds for some $C, \eta, v$ in the sense of Theorem~\ref{thm:LRintro}, then $v\geq v_0$.
\end{thm}

\begin{remark}
As will be described later, $v_0$ is a natural quantity related to the block Jacobi matrix corresponding to this system: it represents the maximum velocity of ballistic transport of a particle in the time evolution given by the block Jacobi matrix.
\end{remark}

As discussed above, the ability to bound the velocity $v$ away from zero corresponds to ballistic transport, while the ability to take $v > 0$ arbitrarily small in a Lieb-Robinson bound corresponds to sub-ballistic transport. Sometimes one can take this one step further and prove a Lieb-Robinson bound with $v = 0$. For example, this was accomplished by Hamza, Sims, Stolz in the case of the isotropic XY chain in random exterior field \cite{HSS12}. To establish such a result they had to show strong dynamical localization for the associated one particle Hamiltonian (the notion of dynamical localization is recalled in more detail in Section \ref{sec:remarks}). Indeed, in the cases considered in \cite{HSS12}, the one particle Hamiltonian turns out to be a random tridiagonal operator with constant off-diagonal entries. With suitable assumptions on the underlying model, the diagonal entries of the one-particle Hamiltonian are shown to be suitably distributed, so that localization follows from the known results in the theory of Anderson localization (for further details, see Section 4.1 in \cite{HSS12}). In general, however, the resulting one particle Hamiltonian is block-diagonal. Localization of random block operators arising in such a way is discussed in \cite{Chapman2015+}, while the fact that localization of the one-particle Hamiltonian implies zero velocity Lieb-Robinsom bound is one of the main results of \cite{HSS12}.

Here we establish the converse of this result and show the following statement (see Theorem~\ref{thm:hss-converse} for the precise meaning of the notion of uniform dynamical localization).

\begin{thm}\label{thm:main-2}
Given a one-dimensional quantum spin chain, if the Lieb-Robinson bound \eqref{eq:lr-spec} holds with $v = 0$, then the associated one particle Hamiltonian exhibits uniform dynamical localization.
\end{thm}

By this result, we see that a zero-velocity Lieb-Robinson bound holds only for very special models as one particle Hamiltonians exhibiting uniform dynamical localization are exceedingly rare. Consequently, the entire regime between sub-ballistic transport and the failure of uniform dynamical localization hides in the seemingly innocuous distinction between the ability to take $v > 0$ arbitrarily small and the ability to take $v = 0$ in a Lieb-Robinson bound.

\bigskip

Returning the the periodic case of primary interest here, the statement of Theorem \ref{thm:main-1} is naturally related to an appropriate ballistic transport statement for the associated one particle Hamiltonian, which is a periodic block Jacobi matrix. Since the appropriate ballistic transport statement for such periodic block Jacobi matrices is also not yet known, we will have to address this issue as well. Thus, the heart of our paper is a comprehensive discussion of ballistic transport for periodic block Jacobi matrices.

Block Jacobi matrices are operators $J : \ell^2(\mathbb{Z})^m \to \ell^2(\mathbb{Z})^m$ of the form
\[
(Ju)_n = a_{n-1}^* u_{n-1} + b_n u_n + a_n u_{n+1}
\]
where $a_n$ and $b_n$ are $m\times m$ complex matrices with $\det a_n \neq 0$ and $b_n^* = b_n$ (so each $u_n \in \mathbb{R}^m$); each $u\in \ell^2(\mathbb{Z})^m$ is viewed as an $\ell^2$ sequence of elements $u_n \in \mathbb{C}^m$. We will assume that $J$ is $q$-periodic, i.e.
\begin{equation}\label{ml05}
a_{n+q} = a_n, \quad b_{n+q} = b_n, \quad \forall n\in\mathbb{Z}.
\end{equation}
For $\psi \in \ell^2(\mathbb{Z})^m$, we are concerned with the time-evolution of
\[
\psi(t) = e^{-itJ} \psi.
\]
Long-standing folklore states that for periodic $J$, $\psi(t)$ should propagate to infinity linearly in time, which is referred to as ballistic motion. The mathematical literature usually describes wave-packet spreading by characterizing the behavior of the $p$-th moments of $\lvert X \rvert$ for $p > 0$ for suitable initial states $\psi$.\footnote{One needs to at least assume finite moments at time zero. If one wants to discuss all moments $p > 0$, one therefore needs to assume at least super-polynomial decay. For simplicity we will focus on exponentially decaying $\psi$ in our discussion of transport exponents.} This is stated in terms of transport exponents
\begin{align*}
\beta_\psi^+(p) & = \limsup_{t \to \infty} \frac{\log \langle \psi(t), \lvert X\rvert^p \psi(t)\rangle }{p \log t}, \\
\beta_\psi^-(p) & = \liminf_{t \to \infty} \frac{\log \langle \psi(t), \lvert X\rvert^p \psi(t)\rangle }{p \log t},
\end{align*}
and their time-averaged counterparts. These exponents are non-decreasing in $p$ and take values in the interval $[0,1]$. Ballistic motion corresponds to exponents being equal to $1$, diffusive transport corresponds to the value $1/2$, while the exponents are $0$ when there is no transport present that is detectable on a power-scale, and in particular when dynamical localization takes place. We refer the reader to \cite{DamanikTcheremchantsev10} for general background.

In this paper, we prove a convergence statement in Hilbert space without any time averaging; this is a stronger characterization of ballistic motion and, as we will see, implies that all transport exponents are equal to 1. It is an extension of work of Asch--Knauf~\cite{AschKnauf98} for Schr\"odinger operators.

We denote by $X$ the (unbounded self-adjoint) operator with
\begin{equation}\label{ml03}
(Xu)_n = n u_n
\end{equation}
and domain
\begin{equation}\label{ml04}
D(X) = \Bigl\{ u \in \ell^2(\mathbb{Z})^m  \Big\vert \sum_{n\in \mathbb{Z}} \lvert n u_n \rvert^2  <\infty  \Bigr\},
\end{equation}
and consider its Heisenberg time evolution $X(t) = e^{itJ} X e^{-itJ}$ with the domain $D(X(t)) = \{ u \in \ell^2(\mathbb{Z})^m \mid e^{-itJ} u \in D(X) \}$ (we will see below that $D(X(t)) = D(X)$ for all $t$).

\begin{thm}\label{Tballistic}
Let $J$ be a periodic block Jacobi matrix as described above.
\begin{enumerate}[{\rm (a)}]
\item There is a bounded self-adjoint operator $Q$ with $\Ker Q = \{ 0\}$ such that for any $\psi \in D(X)$,
\begin{equation}\label{ml60}
\lim_{t\to \infty} \frac 1t X(t) \psi  = Q \psi.
\end{equation}
\item If $f:\mathbb{R} \to \mathbb{R}$ is a bounded continuous function, then
\[
\lim_{t\to \infty} f\left(\frac 1t X(t)\right) \psi = f(Q) \psi
\]
for any $\psi \in \ell^2(\mathbb{Z})^m$.
\item If $\psi \in \ell^2(\mathbb{Z})^m$ decays exponentially,
then for any $p > 0$, $\beta_\psi^\pm(p) = 1$.
\end{enumerate}
\end{thm}

The statement in part (c) has been widely expected to hold in the community and could even be considered a folklore ``result.'' However, to the best of our knowledge this is the first time it is actually proved rigorously, even in the scalar case, $m = 1$. Theorem~\ref{Tballistic} is the key behind the lower bound for the Lieb-Robinson velocity in Theorem~\ref{thm:main-1}. The proof of Theorem~\ref{Tballistic} is given in Section~\ref{sec:ballistic}; and then Theorem~\ref{thm:main-1} is established in Section~\ref{s.xychain}.

As another application of these results, we turn to limit-periodic Schr\"odinger operators, which are operators on $\ell^2(\mathbb{Z})$ of the form $H = \Delta+V$, where $\Delta$ is the discrete Laplacian and $V$ is multiplication by a real-valued limit-periodic sequence, also denoted $V$. Limit-periodic Schr\"odinger operators have been studied extensively and it is known that they can display very diverse spectral properties. Many of those results are stated in terms of denseness or genericity of some spectral property in some space of limit-periodic Schr\"odinger operators; in the usual settings, purely absolutely continuous spectrum is dense \cite{AvronSimon82,DamanikGan11}, purely singular continuous spectrum is generic \cite{Avila09,DamanikGan10}, and purely pure point spectrum is possible \cite{Poschel83,DamanikGan11b}. Here we investigate transport properties of such operators in the same terms. Most notably, we prove generic quasi-ballistic transport for rather general initial states.

\begin{thm}\label{Tgeneric}
For a generic limit-periodic $V$ in the space of all limit-periodic sequences with the $\ell^\infty$ metric, for the discrete Schr\"odinger operator $H=\Delta+V$ we have
\[
\beta^+_\psi(p) = 1
\]
for all $p>0$ and all exponentially decaying $\psi \in \ell^2(\mathbb{Z})\setminus\{0\}$.
\end{thm}

\begin{remark}
As one can infer from the proof, Theorem~\ref{Tgeneric} may also be established in the setting of limit-periodic block Jacobi matrices. We have decided to state and prove it only for discrete Schr\"odinger operators because this is the framework that is most commonly considered, and also because the other results in Section~\ref{Slimper} actually require this more restrictive framework.
\end{remark}

We call this quantum dynamical statement quasi-ballistic transport since the proof establishes ballistic transport only along a subsequence of time scales. Our proof cannot exclude the possibility that $\beta^-_\psi(p) < 1$, or even $\beta^-_\psi(p) = 0$. The latter phenomenon is indeed possible as it is known to occur in a related scenario, namely the super-critical almost Mathieu operator with a very Liouvillean frequency. Indeed in that case, we have $\beta^+_{\delta_0}(p) = 1$ by \cite{Last96} and $\beta^-_{\delta_0}(p) = 0$ by \cite{DamanikTcheremchantsev07}.\footnote{In fact, \cite{DamanikTcheremchantsev07} shows that the time-averaged transport exponents vanish, but this implies the vanishing of the non-time-averaged transport exponents by general principles.} The proof of Theorem~\ref{Tgeneric}, for which Theorem~\ref{Tballistic} is instrumental, can be found in Section~\ref{Slimper}, together with other results on transport (and lack thereof) for limit-periodic Schr\"odinger operators.

\section{Ballistic Transport for Periodic Block Jacobi Matrices}\label{sec:ballistic}

To prove Theorem~\ref{Tballistic}, we will first need some general facts about the time-evolution of $X$ with respect to a bounded block Jacobi matrix. These are given in the following theorem and lemma.

We will define for $N \in \mathbb{N}$ the bounded operators
\[
(X_N u)_n = \begin{cases} -N u_n & n < -N, \\ n u_n & \lvert n \rvert \le N, \\  N u_n & n > N, \end{cases}
\]
and $A_N = i [J, X_N]$,
\[
(A_N u)_n = \begin{cases} - i a_{n-1}^* u_{n-1} + i a_n u_{n+1}   &  \lvert n \rvert \le N-1, \\
- i a_{n-1}^* u_{n-1} & n = N, \\
i a_n u_{n+1}  & n = - N, \\
0 & \lvert n \rvert \ge N+1.
\end{cases}
\]
Note that the $A_N$ converge strongly as $N\to \infty$ to the bounded, self-adjoint operator $A$ given by
\[
(Au)_n =  - i a_{n-1}^* u_{n-1} +  i a_n u_{n+1}.
\]

\begin{thm}\label{Tderivative} Let $J$ be a bounded block Jacobi matrix. For all $T$, $D(X(T)) = D(X)$, and the identity
\begin{equation}\label{ml02}
X(T) \psi = X \psi + \int_0^T A(t)\psi \, dt
\end{equation}
holds for all $\psi \in D(X)$. In particular, for $T>0$,
\[
\lVert X(T) \psi \rVert \le \lVert X \psi \rVert + T \lVert A \rVert \lVert \psi \rVert.
\]
\end{thm}

\begin{proof}
We begin by noting that $X_N(t)$ is an analytic function of $t$, since $J$ is a bounded self-adjoint operator so $e^{-itJ}$ and $e^{itJ}$ are analytic functions of $t$. Differentiating, we have
\[
\frac{d}{dt} X_N(t) = e^{itJ} (iJ) X_N e^{-itJ} + e^{itJ} X_N (-iJ) e^{-itJ} = e^{itJ} A_N e^{-itJ} = A_N(t).
\]
Integrating in $t$, we conclude that for all $T$,
\[
X_N(T) - X_N = \int_0^T A_N(t) dt.
\]
For any $\psi \in D(X)$, strong convergence of $A_N(t)$ to $A(t)$ together with the uniform bound $\lVert A_N(t) \psi \rVert \le 2 \sup_n \lVert a_n \rVert \lVert \psi \rVert$ implies, by dominated convergence, that
\[
 \lim_{N\to \infty} (X_N(T) \psi - X_N \psi ) = \lim_{N\to \infty} \int_0^T A_N(t) \psi \; dt = \int_0^T A(t) \psi \; dt.
\]
The proof is then completed by a double application of the following observation: if $\psi \in D(X(T))$, then $\lim_{N\to \infty} X_N(T) \psi = X(T) \psi$; otherwise, $\lim_{N\to \infty} \lVert X_N(T) \psi \rVert = \infty$. That observation is immediately verified for $T =0$ and extended to any $T$ by conjugation by $e^{-itJ}$.
\end{proof}

For the proof of Theorem~\ref{Tballistic}, we will also need general facts about periodic block Jacobi matrices. It is well-known that, when $J$ is $q$-periodic, it has a direct integral representation when conjugated by the $q$-step Fourier transform. Namely, there is a unitary operator
\[
\mathcal{F} : \ell^2(\mathbb{Z})^m \to L^2 \Big( \partial\mathbb{D}, \frac{d\theta}{2\pi}; (\mathbb{C}^m)^q \Big)
\]
which is first defined by
\[
(\mathcal{F} u)_n (\theta) = \sum_{l\in \mathbb{Z}} u_{n+lq} e^{-il\theta}
\]
for $u \in \ell^1(\mathbb{Z})^m$ and then extended by
\[
\int_0^{2\pi} \lVert \mathcal{F}u_\cdot (\theta) \rVert_2^2 \frac{d\theta}{2\pi} = \sum_{n\in\mathbb{Z}} \lVert u_n \rVert^2.
\]
See \cite[Chapter 5]{Rice} for details in the scalar case ($m=1$), which extend with the same proofs to the block case. It is then straightforward to verify that $J$ and $A$, being $q$-periodic, split into direct integrals with respect to $\mathcal{F}$,
\[
\mathcal{F} J \mathcal{F}^{-1} = \int_0^{2\pi} J_\theta \frac{d\theta}{2\pi}, \qquad \mathcal{F} A \mathcal{F}^{-1} = \int_0^{2\pi} A_\theta \frac{d\theta}{2\pi},
\]
where $J_\theta$ and $A_\theta$ are $mq \times mq$ matrices,
\begin{align}
J_\theta & = \begin{pmatrix}
b_1 & a_1 & & & e^{-i\theta} a_q^* \\
a_1^* & b_2 & a_2 & & \\
& a_2^* & b_3 & \ddots & \\
& & \ddots & \ddots & a_{q-1} \\
e^{i\theta} a_q & & & a_{q-1}^* & b_q
\end{pmatrix}, \nonumber \\ 
A_\theta & = \begin{pmatrix}
0 & i a_1 & & & - i e^{-i\theta} a_q^* \\
-i a_1^* & 0  & i a_2 & & \\
& -i a_2^* & 0  & \ddots & \\
& & \ddots & \ddots & ia_{q-1} \\
i e^{i\theta} a_q & & & -i a_{q-1}^* & 0
\end{pmatrix}. \nonumber
\end{align}
The matrices $J_\theta$ are self-adjoint and depend analytically on $\theta$, so by analytic eigenvalue perturbation theory \cite{RS4}, they are diagonalizable and their eigenvalues are analytic functions of $\theta$; we can label eigenvalues of $J(\theta)$ by distinct analytic functions $\lambda_1(\theta), \dots, \lambda_l(\theta)$, with multiplicities $m_1, \dots, m_l$ independent of $\theta$ (of course, if $\theta$ is such that $\lambda_i(\theta) = \lambda_j(\theta)$ for some $i\neq j$, the corresponding multiplicities are added). Where all the $\lambda_j(\theta)$ are distinct, the corresponding orthogonal projections $P_j(\theta)$ to subspaces $\Ker(J_\theta - \lambda_j(\theta))$ are analytic in $\theta$.

\begin{lemma} \begin{enumerate}[{\rm (a)}]
\item The same value of $\lambda$ cannot be an eigenvalue of $J_\theta$ for more than $2m$ values of $\theta \in [0,2\pi)$.
\item There is a finite set $\mathcal{D} \subset [0,2\pi)$ such that, for all $\theta \in [0,2\pi) \setminus \mathcal{D}$, we have
\[
\lambda_j(\theta) \neq \lambda_k(\theta), \quad \forall j \neq k
\]
and
\[
\frac{\partial \lambda_j(\theta)}{\partial \theta} \neq 0, \quad \forall j.
\]
\end{enumerate}
\end{lemma}

\begin{proof}
(a) If $J_{\theta_j} v(j) = \lambda v(j)$ with the $\theta_j\in [0,2\pi)$ mutually distinct, we construct $u(j) \in \ell^\infty(\mathbb{Z})^m$ by
\[
u(j)_{k+nq} = e^{in\theta_j} v(j)_k, \quad n\in\mathbb{Z}, \quad k\in \{1,2,\dots, q\}.
\]
The $u(j)$ are linearly independent by an argument identical to Lemma 5.3.3 of \cite{Rice}. However, by the construction of $u(j)$, $J u(j) = \lambda u(j)$, which means that $u(j)$ is uniquely defined by $u(j)_1$ and $u(j)_2$ and the recurrence relation
\[
a_{n-1}^* u(j)_{n-1} + (b_n - \lambda I) u(j)_n + a_n u(j)_{n+1} = 0
\]
(here we use $\det a_n \neq 0$). Thus, there cannot be more than $2m$ linearly independent $u(j)$'s.

(b) By (a), all $\lambda_j(\theta)$ are nonconstant analytic functions of $\theta$, so their derivatives have isolated zeros. Similarly, for $j\neq k$, $\lambda_j$ and $\lambda_k$ are distinct analytic functions of $\theta$, so they can only coincide at finitely many points.
\end{proof}

\begin{proof}[Proof of Theorem~\ref{Tballistic}]
The key step is to establish that for every $\theta \in [0,2\pi) \setminus \mathcal{D}$,
\begin{equation}\label{eqnAav}
\lim_{T\to \infty} \frac 1T \int_0^T e^{itJ_\theta} A_\theta e^{-itJ_\theta} dt = q \sum_{j=1}^l \frac{\partial \lambda_j(\theta)}{\partial \theta} P_j(\theta).
\end{equation}
To prove this, let $J_\theta v_j(\theta) = \lambda_j(\theta) v_j(\theta)$ and $J_\theta v_k(\theta) = \lambda_k(\theta) v_k(\theta)$. Then
\begin{align*}
\left \langle v_j(\theta),  \frac 1T \int_0^T e^{itJ_\theta} A_\theta e^{-itJ_\theta} dt \, v_k(\theta) \right\rangle & = \frac 1T \int_0^T \left \langle e^{-itJ_\theta} v_j(\theta),   A_\theta e^{-it J_\theta} v_k(\theta) \right\rangle  dt  \\
& =  \frac 1T \int_0^T e^{it (\lambda_k(\theta)- \lambda_j(\theta))} dt  \; \left \langle v_j(\theta),   A_\theta v_k(\theta) \right\rangle \\
& \to \delta_{j,k} \langle v_j(\theta), A_\theta v_k(\theta) \rangle.
\end{align*}
It therefore suffices to show that
\begin{equation}\label{ml30}
P_j(\theta) A_\theta P_j(\theta) = q \frac{\partial \lambda_j(\theta)}{\partial \theta} P_j(\theta).
\end{equation}
Pick a vector $v_j(\theta) \in \Ker(J_\theta - \lambda_j(\theta))$ which is normalized and analytic in $\theta$. Consider the diagonal matrix
\[
V_\theta = \begin{pmatrix}
1 & & & \\
& e^{i\theta / q} & & \\
& & \ddots & \\
& & & e^{i\theta (q-1) / q}
\end{pmatrix}.
\]
and denote $\tilde v_\theta = V_\theta v_\theta$, $\tilde A_\theta = V_\theta^{-1} A_\theta V_\theta$,  $\tilde J_\theta = V_\theta^{-1} J_\theta V_\theta$. It is easy to compute $\tilde J_\theta$ and $\tilde A_\theta$ explicitly and find
\[
\frac d{d\theta} \tilde J_\theta = \frac 1q \tilde A_\theta.
\]
On the other hand, by normalization,
\[
\left \langle \frac{\partial}{\partial \theta} \tilde v_j(\theta), \tilde v_j(\theta) \right\rangle + \left \langle \tilde  v_j(\theta), \frac{\partial}{\partial \theta} \tilde v_j(\theta) \right\rangle = \frac{\partial}{\partial \theta} \langle\tilde  v_j(\theta), \tilde v_j(\theta) \rangle = 0
\]
so
\[
 \frac{\partial \lambda_j(\theta)}{\partial \theta} = \frac{\partial}{\partial \theta} \langle \tilde v_j(\theta), \tilde J_\theta \tilde v_j(\theta) \rangle =  \langle \tilde v_j(\theta), \frac{\partial  \tilde J_\theta}{\partial \theta} \tilde v_j(\theta) \rangle =  \frac 1q \langle \tilde v_j(\theta), \tilde A_\theta \tilde v_j(\theta) \rangle
\]
This implies \eqref{ml30}, which completes the proof of \eqref{eqnAav}.

(a) By dominated convergence, since the integrand converges for all but finitely many $\theta$, we have convergence in norm
\[
\lim_{T\to \infty} \int_0^{2\pi} \left( \frac 1T \int_0^T e^{itJ_\theta} A_\theta e^{-itJ_\theta} dt \right)  d\theta = \int_0^{2\pi} q \sum_{j=1}^l \frac{\partial \lambda_j(\theta)}{\partial \theta} P_j(\theta) d\theta.
\]
The left hand side is precisely $\lim_{T\to \infty} \frac 1T \int_0^T \mathcal{F} A(t) \mathcal{F}^{-1} dt$, so by undoing conjugation by $\mathcal{F}$, we conclude norm convergence
\[
\lim_{T\to \infty}  \frac 1T \int_0^T A(t) dt = Q
\]
where
\begin{equation} \label{eqnB}
Q =  \mathcal{F}^{-1} \left(  \int_0^{2\pi} q \sum_{j=1}^l \frac{\partial \lambda_j(\theta)}{\partial \theta} P_j(\theta) d\theta \right)  \mathcal{F}.
\end{equation}
Thus, dividing \eqref{ml02} by $T$ and using $\lim_{T\to \infty} \frac 1T X \psi = 0$, we conclude \eqref{ml60}.

(b) Since $X(T) / T$ have a common core on which they converge strongly to $Q$, they also converge to $Q$ in the strong resolvent sense (see Theorem VIII.25 of \cite{RS1}). Since $f$ is a bounded continuous function, we then conclude that $f\left( \frac{X(T)} T\right)$ converge strongly to $f(Q)$, by Theorem VIII.20 of \cite{RS1}.

(c) By a result of Damanik--Tcheremchantsev \cite[Theorem 2.22]{DamanikTcheremchantsev10}, $\beta_\psi^\pm(p) \le 1$ for all $p>0$, so it remains to prove the opposite inequality.

Take the continuous bounded function $g_C:\mathbb{R} \to \mathbb{R}$ defined by
\[
g_C(x) = \min( \lvert x \rvert , C )^p.
\]
By part (b), for any $\psi$,
\[
\lim_{T\to \infty} \left\langle \psi,  g_C\left( \frac{X(T)} T\right)  \psi \right\rangle = \langle \psi, g_C(Q) \psi \rangle.
\]
Since $g_C\left( \frac{X(T)} T\right)$ is increasing in $C$ and converges to $\left\langle \psi,  \left| \frac{X(T)} T\right|^p  \psi \right\rangle$, we conclude that
\[
\liminf_{T\to \infty}  \frac 1{T^p} \langle \psi,  \lvert X(T)\rvert^p  \psi \rangle    \ge \langle \psi, \lvert Q \rvert^p \psi \rangle > 0.
\]
This implies $\beta_\psi^-(p) \ge 1$, concluding the proof.
\end{proof}

As our final topic in this section, we prove a corollary of this method which will be used later in the paper. A notational remark: for this statement, we want to de-emphasize the block structure and view our Hilbert space as $\ell^2(\mathbb{Z})$. Then we can refer to the vectors $\delta_n \in \ell^2(\mathbb{Z})$ as usual. Note that with this notation, for instance,
\begin{equation}\label{Xm}
X \delta_{n} = \left\lfloor \frac nm \right\rfloor \delta_n.
\end{equation}

\begin{cor}\label{cor:lr}
For any $\epsilon > 0$, there exist constants $T_0, \tilde C>0$ and $K \in \mathbb{Z}$, such that for all $T\ge T_0$, there exist $n, k\in \mathbb{Z}$ with
\[
m(\lVert Q \rVert - \epsilon) T \le \lvert n\rvert \le m\lVert Q \rVert T + m-1
\]
and $\lvert k \rvert \le K$ such that
\[
 \lvert \langle \delta_n, e^{-iTJ} \delta_k \rangle \rvert^2 \ge \frac{\tilde C}{T}.
\]
\end{cor}

\begin{proof} Let us denote $v_0 = \lVert Q\rVert$. For any $\epsilon >0$, the spectral projection of $Q$ to $[-v_0, - v_0 + \epsilon/2] \cup [v_0 - \epsilon/2, v_0]$ is nontrivial, so there exists a vector $\psi$ such that
\[
\chi_{[-v_0, -v_0 + \epsilon/2] \cup [v_0 - \epsilon/2, v_0 ] }(Q) \psi  \neq 0
\]
and, by approximation, the vector $\psi$ can be chosen to be of bounded support, supported on $[-K,K] \cap \mathbb{Z}$ for some $K\in\mathbb{Z}$.

Let us assume that
\[
\chi_{[v_0 - \epsilon/2, v_0] }(Q) \psi  \neq 0;
\]
the other case $\chi_{[-v_0, -v_0 + \epsilon/2] }(Q) \psi  \neq 0$ is treated analogously.

Let $j:\mathbb{R} \to \mathbb{R}$ be a continuous function with $0 \le j \le 1$, $j=1$ on $[v_0 - \epsilon/2, v_0]$, and $j=0$ outside of $[v_0 - \epsilon, v_0+\epsilon]$.

By Theorem~\ref{Tballistic}(b),
\[
j\left(\frac 1T X(T)\right) \psi \to j(Q) \psi.
\]
In particular, since
\[
\lVert j(Q) \psi \rVert \ge \lVert\chi_{[v_0 - \epsilon/2, v_0] }(Q) \psi \rVert > 0,
\]
there exists a constant $C_1>0$ such that, for all $T\ge T_0$,
\[
\left\lVert \chi_{[v_0-\epsilon,v_0]}\left(\frac 1T X(T)\right)\psi \right\rVert^2 \ge \lVert j(Q) \psi \rVert \ge C_1.
\]
Keeping \eqref{Xm} in mind, this implies
\[
\sum_{n = m \lceil (v_0-\epsilon) t\rceil}^{m\lfloor v_0 t \rfloor +m-1} \lvert \langle \delta_n, e^{-itJ} \psi \rangle \rvert^2 \ge C_1.
\]
By the Cauchy--Schwarz inequality, this implies
\[
\sum_{n = m \lceil (v_0-\epsilon)t\rceil}^{m\lfloor v_0 t \rfloor+m-1} \sum_{k = -K}^K \lvert \psi_k \rvert^2  \lvert \langle \delta_n, e^{-itJ} \delta_k \rangle \rvert^2 \ge C_1.
\]

Since the sum has at most $(mv_0 T - m(v_0 -\epsilon)T +m)(2K+1)$ terms, we conclude that some term in the sum obeys
\[
\lvert \psi_k \rvert^2 \lvert \langle \delta_n, e^{-itJ} \delta_k \rangle \rvert^2 \ge \frac{C_1}{m(v_0 T-(v_0-\epsilon) T+1)(2K+1)} \ge \frac {C_2}T.
\]
Using $\lvert \psi_k \rvert^2 \le \lVert \psi\rVert^2$, this completes the proof with $\tilde C = C_2 / \lVert \psi \rVert^2$.
\end{proof}

\section{The Anisotropic $XY$ Chain}\label{s.xychain}

In this section, among other results, we prove Theorems \ref{thm:main-1} and \ref{thm:main-2}. Indeed, these theorems are implied by the Theorems \ref{thm:v-bound} and \ref{thm:hss-converse}, respectively, from this section.

\subsection{The Heisenberg Evolution and the Lieb-Robinson Bound}

Take two regions $S_1, S_2\subset \Lambda = [m, n]\cap \Z$, such that $S_1\cap S_2 = \emptyset$. Let $A \in \mathcal{A}_{S_1}$, and $B\in \mathcal{A}_{S_2}$. Observe then that $A$ and $B$ commute as operators in $\mathcal{A}_{\Lambda}$. Now let us consider the Heisenberg evolution of the operator $A$ via
\begin{align}\label{eq:evol}
  \tau_t^{(\Lambda)}(A): \R\ni t\mapsto e^{it H^{(\Lambda)}}Ae^{-it H^{(\Lambda)}}.
\end{align}

Since, due to the local interaction in the chain, perturbations in $S_1$ by $A$ will propagate through the chain, one expects the commutator $[\tau_t^{(\Lambda)}(A), B]$ to be nonzero for some values of $t$. The norm of the commutator, which we denote by $\mathbf{P}_t^{(\Lambda)}(A, B)=\norm{[\tau_t^{(\Lambda)}(A), B]}$, becomes an indicator of propagation of local disturbances.

The Lieb-Robinson bound puts a restriction on the growth of $\mathbf{P}_t^{(\Lambda)}(A, B)$ \cite{LR72}, according to which we have
\begin{align}\label{eq:lr}
  \mathbf{P}_t^{(\Lambda)}(A, B)\leq C\norm{A}\norm{B}e^{-\eta(d(S_1, S_2)-v\abs{t})},
\end{align}
with $C>0$, $\eta > 0$, and $v\geq 0$ constant, and all three constants do not depend on the length of the lattice $\Lambda$. Here $d(S_1, S_2)$ denotes the distance between the two regions:
\begin{align*}
d(S_1, S_2) = \min_{(i,j)\in S_1\times S_2}\set{\abs{i-j}}.
\end{align*}
Notice that for $\abs{t}\ll d(S_1,S_2)$, $\mathbf{P}_t^{(\Lambda)}(A, B)$ remains exponentially small in the variable $d(S_1, S_2)$. The constant $C$ may depend on the regions $S_1$ and $S_2$ (such as their sizes), the norms of $A$ and $B$, the lattice structure, and the interaction (see Section 2 in \cite{NS09}). The velocity, $v$, depends on the interaction and the lattice structure only. Typically, $\eta$ is taken as a fixed constant, but may be allowed to vary depending on $d(S_1, S_2)$. Thus the Lieb-Robinson bound is a statement about a relationship among the three variables, $C$, $\eta$ and $v$, such that the bound \eqref{eq:lr} holds independently of the lattice size (and, typically, independently of $d(S_1, S_2)$).

Evidently there are many potential choices of the three variables for a given system that would satisfy \eqref{eq:lr}. Hence much effort has been put into establishing the sharpest bounds on the variables, for any given system, subject to \eqref{eq:lr} (see, for example, \cite{NS10}). For instance, for the Hamiltonian \eqref{eq:haml} with random nearest neighbor interaction (and/or the external field), under some technical conditions, with a suitable choice of $C > 0$ and $\eta > 0$ \textit{constant}, one can take $v = 0$. This, in some sense, is the strongest type of bound, defined in \cite{HSS12} as the \textit{zero velocity Lieb-Robinson bound}. We elaborate on this later. In fact, we show that zero velocity bounds cannot hold for a large class of systems \eqref{eq:haml} (the class being determined by the type of interaction and the external field), namely those where a certain dynamical localization phenomenon is impossible. The periodic case is the primary example, but there are many more examples to which this observation can be applied.

Let us mention, in particular since we shall need this later, that as a special case of Theorem 1 in \cite{NS09}, in the case of \eqref{eq:haml}, the variables $C$ and $\eta$ may be chosen to be constant (i.e. independent of the lattice size, $d(S_1, S_2)$, and even the geometry of $S_1$ and $S_2$) as long as one restricts to a suitable (and quite natural) class of observables. Namely, we restrict to those pairs of observables $(A, B)$ whose supports do not interlace; that is, if $A$ is supported on $S_1$ and $B$ is supported on $S_2$, then we require that $\max\set{s\in S_1} < \min\set{s\in S_2}$. Let us denote the class of pairs of such observables by $\mathcal{D}^{(\Lambda)}$, where the superscript emphasizes the lattice:
\begin{align*}
\mathcal{D}^{(\Lambda)}:=\set{(A, B)\in \mathcal{A}_{\Lambda}\times\mathcal{A}_{\Lambda}: \supp(A)\hspace{2mm} \text{and}\hspace{2mm} \supp(B)\hspace{2mm}\text{do not interlace}}.
\end{align*}
 We then have

\begin{thm}[Nachtergaele-Sims]\label{thm:ns}
There exist constants $C, \eta, v > 0$ independently of $\Lambda$ such that for any choice of $(A, B)\in\mathcal{D}^{(\Lambda)}$ with $\supp(A) = S_1$ and $\supp(B) = S_2$,
\begin{align*}
\mathbf{P}_t^{(\Lambda)}(A, B)\leq C\norm{A}\norm{B}e^{-\eta(d(S_1, S_2) -v\abs{t})}.
\end{align*}
\end{thm}

A system like \eqref{eq:haml} can be transformed to a tight-binding model via the well known Jordan-Wigner transformation, and solved by the Lieb-Schultz-Mattis method \cite{LSM61}. For the reader's convenience and for completeness, we provide a brief summary below.

\subsection{Transformation to a Tight-Binding Hamiltonian}

For what follows, a self-contained treatment applicable to the model \eqref{eq:haml} has been given in \cite{HSS12}, with references therein for the more extensive discussion. In what follows, we take $\Lambda = [m, n]\cap\Z$ as above.

For $j\in \Lambda$, define the raising and the lowering operators, $a_j^*$ and $a_j$, respectively, by
\begin{align*}
a_j^*=\frac{1}{2}(\sigma_j^{(x)}+i\sigma_j^{(y)})\hspace{2mm}\text{ and }\hspace{2mm}a_j=\frac{1}{2}(\sigma_j^{(x)}-i\sigma_j^{(y)}).
\end{align*}
Define also the operators $c_j^*$ and $c_j$ on the lattice $\Lambda$ by
\begin{align*}
c_m^*=a_m^*\hspace{2mm}\text{ and }c_{m+j}^*=\sigma_m^{(z)}\cdots\sigma_{m+j-1}^{(z)}a_{m+j}^*\hspace{2mm}\text{ for }\hspace{2mm} j\geq 1, m+j \leq n,
\end{align*}
and $c_j$ is defined similarly with $a_j$ in place of $a_j^*$. Notice that at any given node $j\in\Lambda$, the algebra $\mathcal{A}_{\set{j}}$ is generated by $\set{a_j, a_j^*, a_ja_j^*, a_j^* a_j}$. At the same time, $\sigma_j^{(z)}$ can be expressed as
\begin{align}\label{eq:sigma}
\sigma_j^{(z)} = 2a_j^*a_j - \mathbb{I} = 2c_j^*c_j - \mathbb{I}.
\end{align}
Thus observables in $\mathcal{A}_{\set{j}}$ can be expressed as polynomials in the operators $c_j$ and $c_j^*$. Now define
\begin{align}\label{eq:c-vect}
\begin{split}
&\mathcal{C}^{(\Lambda)}:=({c_m},{c_m^*}, \dots,{c_n}, {c_n^*}),\\
&\tau_t^{(\Lambda)}(\mathcal{C}^{(\Lambda)}):= ({\tau_t^{(\Lambda)}(c_m)}, {\tau_t^{(\Lambda)}(c_m^*)},\dots,{\tau_t^{(\Lambda)}(c_n)},{\tau_t^{(\Lambda)}(c_n^*)}).
\end{split}
\end{align}
Then we have \cite{HSS12}
\begin{align}\label{eq:evol-schro}
\tau_t^{(\Lambda)}(\mathcal{C}^{(\Lambda)})=e^{-it M}\mathcal{C}^{(\Lambda)},
\end{align}
where the matrix $M$ is block-diagonal, given by
\begin{align}\label{eq:M}
M=\begin{pmatrix}
J_m & \Gamma_{m+1} & 0 & \cdots & 0\\
\Gamma_{m+1}^* & J_{m+1} & \ddots & \ddots & \vdots\\
0 & \ddots &  \ddots & \ddots & 0\\
\vdots & \ddots & \ddots & \ddots & \Gamma_{n}\\
0 & \cdots & 0 & \Gamma_{n}^* & J_n\\
\end{pmatrix}
\end{align}
with the matrices $J_{m\leq j \leq n}$ and $\Gamma_{m+1\leq j\leq n}$ given by
\begin{align}\label{eq:blocks}
J_j=2\begin{pmatrix}
\nu_j & 0\\
0 & -\nu_j
\end{pmatrix}
\hspace{2mm}\text{ and }\hspace{2mm}
\Gamma_j=2\begin{pmatrix}
-\mu_{j-1} & -\mu_{j-1}\gamma_{j-1}\\
\mu_{j-1}\gamma_{j-1} & \mu_{j-1}
\end{pmatrix},
\end{align}
and $\Gamma^*_j$ is the transpose of $\Gamma_j$. We write $M^{(\Lambda)}$ for $M$ whenever the lattice size $\Lambda$ needs to be emphasized.

\begin{remark}
In \cite{HSS12}, the evolution \eqref{eq:evol-schro} is given with $e^{-2itM^{(\Lambda)}}\mathcal{C}^{(\Lambda)}$ on the right; for convenience, we have factored $2$ into $J_j$ and $\Gamma_j$ in \eqref{eq:blocks}.
\end{remark}

Notice the dimension reduction from the original $2^{\abs{\Lambda}}$-dimensional Hilbert space, the $\abs{\Lambda}$-fold \textit{tensor product} $\bigotimes_{i=m}^n\C^2$, to the $2\abs{\Lambda}$-dimensional Hilbert space, the $\abs{\Lambda}$-fold \textit{direct sum} $\bigoplus_{i=m}^n \left(\mathcal{A}_{i}\oplus\mathcal{A}_{i}\right)$.

\subsection{General Upper Bounds}

It has been suggested \cite{BO07, ZPP08} that disorder in the couplings $\set{\mu_j}$ and $\set{\nu_j}$ may cause localization in the sense of exponential decay of $\sup_{t} \mathbf{P}_t^{(\Lambda)}(A, B)$ in the variable $d(S_1, S_2)$. In \cite{HSS12} this was confirmed rigorously:

\begin{thm}[Hamza-Sims-Stolz]\label{thm:h-s-s}
Suppose that $M^{(\Lambda)}$ is dynamically localized in the sense that there exist constants $C > 0$ and $\eta > 0$ such that
\begin{align*}
\mathbb{E} \Big( \sup_{t\in\R}\abs{M^{(\Lambda)}_{j, k}(t)}\Big) \leq Ce^{-\eta\abs{j-k}},
\end{align*}
where $M_{j, k}^{(\Lambda)}(t)$ is the $(j, k)$ entry of the matrix $e^{-it M^{(\Lambda)}}$ and $\mathbb{E}(\cdot)$ is the expectation with respect to a given probability measure, with respect to which the sequences $\set{\mu_j}$ and $\set{\nu_j}$ are drawn. Then with $m_*, m^*\in \Lambda$, $m_* < m^*$, for any $A \in \mathcal{A}_{\set{m_*}}$ and $B\in\mathcal{A}_{[m^*, n]}$, $\mathbf{P}_t^{(\Lambda)}(A, B)$ satisfies the zero velocity Lieb-Robinson bound in expectation, namely
\begin{align}\label{eq:sLR}
\mathbb{E}\left(\sup_{t\in\R}\mathbf{P}_t^{(\Lambda)}(A, B)\right)\leq C'\norm{A}\norm{B}e^{-\eta(m^*-m_*)},
\end{align}
where $C' > 0$ is independent of the lattice size $n$.
\end{thm}

As remarked in \cite[Remark 3.3]{HSS12}, the theorem applies to more general observables. Indeed, a general upper bound on $\mathbf{P}_t^{(\Lambda)}(A,B)$ is given in Proposition \ref{prop:hss-upper} below (we follow the same techniques as in \cite{HSS12}, thus Proposition \ref{prop:hss-upper} is a straightforward extension). Also notice that in the statement of the theorem, the pair of observables $(A, B)$ belongs to $\mathcal{D}^{(\Lambda)}$.

On the other hand, dynamical localization is known from the Anderson localization theory for a family of disordered matrices, including, for example, $M$ with $\gamma_j\equiv 0$ (which corresponds to the isotropic XY chain), $\mu_j = \mu$ constant and $\nu_j$ being i.i.d. whose common distribution is absolutely continuous with bounded and compactly supported density (see the discussion in Section 4.1 of \cite{HSS12}).

Following the technique from the proof of \cite[Theorem 3.2]{HSS12}, we can give the following upper bound on $\mathbf{P}^{(\Lambda)}_t(A, B)$ with $(A, B)\in\mathcal{D}^{(\Lambda)}$ and, moreover, with the assumption that $A$ is \textit{decomposable}. That is, we assume that given that the support of $A$ is $S_1 = \set{s_1 < s_2 < \cdots < s_k}$, $A$ can be written as
\begin{align}\label{eq:decomp}
A = \left(\bigotimes_{s\in S_1}A_{s}\right)\otimes\mathbb{I}_{\Lambda\setminus S_1},
\end{align}
where for $s\in S_1$, $A_s$ is an observable on the local Hilbert space $\C^2$ on site $s$.

\begin{prop}\label{prop:hss-upper}
There exists a constant $C > 0$ independently of $\Lambda$, such that for all pairs of observables $(A, B)\in\mathcal{D}^{(\Lambda)}$ with $A$ supported on $S_1$ and $B$ supported on $S_2$, and $A$ being decomposable as in \eqref{eq:decomp},
\begin{align}\label{eq:prop-eq}
\begin{split}
\mathbf{P}_t^{(\Lambda)}(A, B) &\leq C\norm{A}\norm{B}\sum_{s\in S_1}\sum_{k\leq s}\sum_{k'\geq r}\abs{M_{k, k'}^{(\Lambda)}(t)}\\
&\leq C\abs{S_1}\norm{A}\norm{B}\sum_{k\leq l}\sum_{k'\geq r}\abs{M_{k,k'}^{(\Lambda)}(t)},
\end{split}
\end{align}
where $\abs{S_1}$ is the cardinality of $S_1$, $l=\max\set{s\in S_1}$, and $r = \min\set{s\in S_2}$ {\rm (}we assume that $l < r${\rm )}.
\end{prop}

\begin{proof}
For $s\in S_1$, define $\tilde{A}_s$ as $A_s\otimes \mathbb{I}_{\Lambda\setminus\set{s}}$. We claim
\begin{align}\label{eq:claim-eq}
\mathbf{P}_t^{\Lambda}(A, B)\leq \sum_{s\in S_1}\left(\prod_{j\neq s}\norm{A_j}\right)\norm{[\tau_t^{(\Lambda)}(\tilde{A}_s), B]}.
\end{align}

Notice that $A = \prod_{s\in S_1}\tilde{A}_s$, where the product here means composition of operators (it does not matter in which order the composition is taken, since the operators $\tilde{A}_{s_1}$, $\tilde{A}_{s_2}$ commute for any $s_1, s_2\in S_1$). We now proceed by induction on the cardinality of $S_1$. The case $S_1 = \set{s_1}$ is trivial. Assume $S_1 = \set{s_1, s_2}$. By the Leibniz rule and the automorphism property of $\tau_t^{(\Lambda)}(\cdot)$, we can write
\begin{align*}
[\tau_t^{(\Lambda)}(\tilde{A}_{s_1}\tilde{A}_{s_2}), B] = \tau_t^{(\Lambda)}(\tilde{A}_{s_1})[\tau_t^{(\Lambda)}(\tilde{A}_{s_2}), B] + [\tau_t^{(\Lambda)}(\tilde{A}_{s_1}), B]\tau_t^{(\Lambda)}(\tilde{A}_{s_2}).
\end{align*}
Now the desired inequality follows from application of the triangle inequality. To complete the induction, for any $k\geq 2$ we can replace $\tilde{A}_{s_2}$ above with $\tilde{A}_{s_k}$, and $\tilde{A}_{s_1}$ with $\prod_{s\in S_1\setminus\set{s_k}}\tilde{A}_{s}$, and proceed by repeated application of the Leibniz rule. Thus it remains to give a suitable bound on $\norm{[\tau_t^{(\Lambda)}(\tilde{A}_s), B]}$ for every $s\in S_1$.

As has been mentioned above,
\begin{align*}
\beta=\set{a_s, a_s^*, a_s^*a_s, a_sa_s^*}
\end{align*}
forms a basis for $\mathcal{A}_{\set{s}}$. Notice that a bound on $\norm{[\tau_t^{(\Lambda)}(a_s^*a_s), B]}$ as well as $\norm{[\tau_t^{(\Lambda)}(a_sa_s^*), B]}$ can be given in terms of $\norm{[\tau_t^{(\Lambda)}(a_s^*), B]}$ and $\norm{[\tau_t^{(\Lambda)}(a_s), B]}$ by using the Leibniz rule as above. A bound on $\norm{[\tau_t^{(\Lambda)}(a_s^*), B]}$ (and, similarly, on $\norm{[\tau_t^{(\Lambda)}(a_s), B]}$) is given in the proof of \cite[Theorem 3.2]{HSS12}. For completeness, we repeat the computations here for $\norm{[\tau_t^{(\Lambda)}(a_s), B]}$ (the case with $a_s^*$ in place of $a_s$ is handled similarly).

With $\mathcal{C}^{(\Lambda)}$ from \eqref{eq:c-vect}, denote by $\mathcal{C}^{(\Lambda),(j)}$ the $j$-th component of $\mathcal{C}^{(\Lambda)}$. From \eqref{eq:evol-schro}, we obtain, for any $s\in S_1$,
\begin{align}\label{eq:c-right-bound}
\norm{[\tau_t^{(\Lambda)}(c_s), B]} = \norm{\sum_{j} M^{(\Lambda)}_{\tilde{s}, j}(t)[\mathcal{C}^{(\Lambda),(j)}, B]} \leq 2\norm{B}\sum_{j}\abs{M_{\tilde{s}, j}^{(\Lambda)}(t)},
\end{align}
where $j$ runs through all the column indices of $M^{(\Lambda)}$ and $\tilde{s} = m + 2i$ provided that $s = m + i$ (the bound on the right follows after an application of the triangle inequality followed by the observation that $\norm{\mathcal{C}^{(\Lambda),(j)}}=1$, and the observation that for all $j < r$, $[\mathcal{C}^{(\Lambda),(j)}, B] = 0$ since $\supp (\mathcal{C}^{(\Lambda), (j < r)})\subset[m, m+r-1]$). Since $a_s = \sigma_m^{(z)}\cdots \sigma_{s-1}^{(z)}c_{s}$ and $\norm{\sigma_{s}^{(z)}} = 1$ for all $s\in\Lambda$, we have, after a repeated application of the Leibniz rule,
\begin{align}\label{eq:a-right-bound}
\norm{[\tau_t^{(\Lambda)}(a_s), B]}\leq \norm{[\tau_t^{(\Lambda)}(c_s), B]} + \sum_{j\leq s-1}\norm{[\tau_t^{(\Lambda)}(\sigma_j), B]}.
\end{align}
Finally, using \eqref{eq:sigma} and an application of the Leibniz rule, we get
\begin{align}\label{eq:sigma-right-bound}
\norm{[\tau_t^{(\Lambda)}(\sigma_j), B]}\leq 2\left(\norm{[\tau_t^{(\Lambda)}(c_j), B]}+\norm{[\tau_t^{(\Lambda)}(c_j^*), B]}\right),
\end{align}
where for all $s\in\Lambda$, $\norm{[\tau_t^{(\Lambda)}(c_s^*), B]}$ is bounded similarly to $\norm{[\tau_t^{(\Lambda)}(c_s), B]}$, as in \eqref{eq:c-right-bound}, with $\tilde{s} = m + 2i + 1$ when $s = m + i$. Putting \eqref{eq:c-right-bound}, \eqref{eq:a-right-bound}, and \eqref{eq:sigma-right-bound} together, we get
\begin{align}\label{eq:a-bound}
\norm{[\tau_t^{(\Lambda)}(a_s), B]}\leq 8\norm{B}\sum_{k\leq s}\sum_{k'\geq r}\abs{M_{k, k'}^{(n)}(t)}.
\end{align}
Now, after writing $\mathcal{A}_s$, $s\in S_1$, as a linear combination in the basis $\beta$ and applying the Leibniz rule to $[\tau_t^{(\Lambda)}(a_sa_s^*), B]$ and $[\tau_t^{(\Lambda)}(a_s^*a_s), B]$, \eqref{eq:a-bound} together with \eqref{eq:claim-eq} implies \eqref{eq:prop-eq}.
\end{proof}

\begin{remark}\label{rem:zero-vel-ext}
A few remarks are in order here.

\begin{enumerate}

\item Notice that Theorem \ref{thm:h-s-s} now extends to all pairs of observables $(A, B)\in \mathcal{D}^{(\Lambda)}$ with decomposable $A$. The price to pay here is, in the worst case, the extra factor $\abs{\supp(A)}$ in the Lieb-Robinson bound.

\item The assumption of decomposability of $A$ can be dropped; the price to pay is, in the worst case, an additional factor (in addition to $\abs{S_1}$) of $4^{\abs{S_1}}$ in the Lieb-Robinson bound. This factor comes from the fact that the local algebra is four-dimensional, thus the tensor product of local algebras over the sites in $S_1$ is $4^{\abs{S_1}}$ dimensional. Hence each operator supported on $S_1$ can be written, in the worst case, as a linear combination of $4^{\abs{S_1}}$ decomposable operators.

\item Based on the proof of Proposition \ref{prop:hss-upper}, one could tighten the bounds depending on the specific nature of the given observable $A$. For example, when $A$ is a linear combination of the observables $c_l$, $c_l^*$, $c_lc_l^*$ and $c_l^*c_l$, $l\in\Lambda$, a bound of the type \eqref{eq:c-right-bound} holds.

\item Finally, as a special case, it is easy to see that when $A\in \mathcal{A}_{\set{l}}$, we in fact have
\begin{align*}
\mathbf{P}_t^{(\Lambda)}(A, B)\leq C\norm{A}\norm{B}\sum_{k\leq l}\sum_{k'\geq r}\abs{M_{k,k'}^{(n)}(t)},
\end{align*}
which is precisely the bound presented in the proof of \cite[Theorem 3.2]{HSS12}.
\end{enumerate}
\end{remark}

\subsection{Lower Bounds on the Lieb-Robinson Velocity}
In view of the recent developments, the question of the existence of a lower bound on propagation arises naturally. In this section we give a lower bound on the Lieb-Robinson velocity.

So far we have been working on finite and fixed lattices $\Lambda = [m, n]\cap \Z$. In what follows, we continue to work on finite but in general not fixed lattices (in fact, in Theorem \ref{thm:v-bound} we will be taking the limit $\Lambda\rightarrow\Z$). Notice that the definition of $M^{(\Lambda)}$ in \eqref{eq:M} is such that for every $\Lambda = [m, n]\cap\Z$, $M^{(\Lambda)}$ can be viewed as the truncation of $M^{(\infty)}$, the infinite matrix constructed in the same way as ${M}^{(\Lambda)}$ with $m\rightarrow\ -\infty$ and $n\rightarrow\infty$. Then, if we denote by $M^{(\infty)}(t)$ the matrix $e^{-itM^{(\infty)}}$, we see that for every fixed $t\in\R$, $M^{(\Lambda)}(t)\rightarrow M^{(\infty)}(t)$ in the strong operator topology (we view $M^{(\Lambda)}$ as an operator on $\ell^2(\Z)$ after extending it to the infinite matrix by padding it with zero rows and columns). For this reason, to avoid ambiguity when taking limits, it is convenient to introduce the following convention. Let us take
\begin{align*}
\mathcal{C}:=(\dots, c_{-1}, c_{-1}^*, c_{0}, c_0^*, c_1, c_1^*\dots),
\end{align*}
such that for any $j$, $c_j = \langle \delta_{2j-1}, \mathcal{C}\rangle$, and $c_j^* = \langle \delta_{2j}, \mathcal{C}\rangle$, with $\set{\delta_j}_{j\in\Z}$ being the canonical basis of $\ell^2(\Z)$. Then for $\Lambda = [m, n]\cap \Z$, we can view $M^{(\Lambda)}(t)$ as an operator acting on $\mathcal{C}$ such that, for every $j\in \Lambda$, we have
\begin{align*}
\tau_t^{\Lambda}(c_j) = \langle \delta_{2j-1}, M^{(\Lambda)}(t)\mathcal{C} \rangle\hspace{2mm}\text{ and }\hspace{2mm}\tau_t^{\Lambda}(c_j^*)=\langle \delta_{2j}, M^{(\Lambda)}(t)\mathcal{C}\rangle.
\end{align*}

We proceed with this convention in mind.

\begin{lemma}\label{thm:propagator-lower}
For any entry $m^{(\Lambda)}(t)$ of the matrix $M^{(\Lambda)}(t)$ above the main diagonal, there exists $(A, B)\in\mathcal{D}^{(\Lambda)}$ with $\norm{A} = \norm{B} = 1$, such that $\mathbf{P}_t^{(\Lambda)}(A, B) \geq \abs{m^{(\Lambda)}(t)}$ for every $t\in \R$. Furthermore, if $m^{(\Lambda)}(t)$ is the $(\tilde{l}, \tilde{r})$ entry of the matrix $M^{(\Lambda)}(t)$, then we can choose $(A, B)$ above such that if $l = \max\set{s\in\supp(A)}$ and $r = \min\set{s\in\supp(B)}$, then $\abs{r - l}\geq \frac{\abs{\tilde{r}-\tilde{l}}}{2}-1$.
\end{lemma}

\begin{proof}
Provided that $m^{(\Lambda)}(t)$ is the $(\tilde{l}, \tilde{r})$ entry of $M^{(\Lambda)}(t)$ with $\tilde{r} > \tilde{l}$, we need to handle the following four cases, with $l, r\in \Lambda$, $l < r$.

\begin{enumerate}

\item $\tilde{l} = 2l-1$ and $\tilde{r} = 2r-1$.

\item $\tilde{l} = 2l-1$ and $\tilde{r} = 2r$.

\item $\tilde{l} = 2l$ and $\tilde{r} = 2r-1$.

\item $\tilde{l} = 2l$ and $\tilde{r} = 2r$.

\end{enumerate}

We begin with (1). In this case, let us take $(A, B) = (c_l, a_r^*)$. From \eqref{eq:evol-schro} we have, with $\mathcal{C}^{(j)}$ being the $j$th element of $\mathcal{C}$,
\begin{align*}
\mathbf{P}_t^{(\Lambda)}(c_l, a_r^*) = \norm{\left[\sum_{j\hspace{1mm}\mathrm{ odd}} M_{\tilde{l}, j}^{(\Lambda)}(t)\mathcal{C}^{(j)}+ \sum_{j\hspace{1mm}\mathrm{ even}} M_{\tilde{l}, j}^{(\Lambda)}(t)\mathcal{C}^{(j)},\hspace{1mm} a_r^*\right]}.
\end{align*}
Observe that for all $j < r$, $a_r^*$ commutes with $\mathcal{C}^{(2j-1)}$ and $\mathcal{C}^{(2j)}$; also, $a_r^*$ commutes with $\mathcal{C}^{(2r)}$. Thus we can write
\begin{align*}
&\mathbf{P}_t^{(\Lambda)}(c_l, a_r^*) \\ &= \norm{\left[\sum_{j > r} M_{\tilde{l}, 2j-1}^{(\Lambda)}(t)\mathcal{C}^{(2j-1)}, a_r^*\right]+ \left[\sum_{j > r} M_{\tilde{l}, 2j}^{(\Lambda)}(t)\mathcal{C}^{(2j)}, a_r^*\right] + \left[M_{\tilde{l}, \tilde{r}}^{(\Lambda)}(t)\mathcal{C}^{(\tilde{r})}, a_r^*\right]}.
\end{align*}
Notice that for each $i\in\Lambda$, $\sigma_i^{(z)}v=v$ with $v=\bigotimes_\Lambda \left(\begin{smallmatrix} 1\\ 0\end{smallmatrix}\right)$. Thus, since $a_r^*\left(\begin{smallmatrix}1\\0\end{smallmatrix}\right)=0$, we have (recall that for $j\in\Lambda$, $\mathcal{C}^{(2j)}$ and $\mathcal{C}^{(2j-1)}$ are defined by $\left(\prod_{i < j}\sigma_i^{(z)}\right)a_j^*$ and $\left(\prod_{i < j}\sigma_i^{(z)}\right)a_j$, respectively, where the product denotes composition of operators)
\begin{align*}
\left(\left[\sum_{j > r}M_{\tilde{l},2j-1}^{(\Lambda)}(t)\mathcal{C}^{(2j-1)}, a_r^*\right] + \left[\sum_{j > r} M_{\tilde{l},2j}^{(\Lambda)}(t)\mathcal{C}^{(2j)}, a_r^*\right]\right)v = 0,
\end{align*}
and
\begin{align*}
[M_{\tilde{l}, \tilde{r}}^{(\Lambda)}(t)\mathcal{C}^{(\tilde{r})}, a_r^*]v = M_{\tilde{l}, \tilde{r}}^{(\Lambda)}(t)v.
\end{align*}
Thus we have
\begin{align*}
\mathbf{P}_t^{(\Lambda)}(c_l, a_r^*)\geq \norm{[M_{\tilde{l},\tilde{r}}^{(\Lambda)}(t)\mathcal{C}^{(\tilde{r})}, a_r^*]v}=\norm{M_{\tilde{l},\tilde{r}}^{(\Lambda)}(t)v}=\abs{M_{\tilde{l},\tilde{r}}^{(\Lambda)}(t)}.
\end{align*}

Now the cases (2), (3) and (4) are handled similarly by taking $(A, B)$ to be $(c_l, a_r)$, $(c_l^*, a_r)$, and $(c_l^*, a_r^*)$, respectively; in the cases (2) and (3) the vector $v$ above should be $\bigotimes_\Lambda\left(\begin{smallmatrix}0\\1\end{smallmatrix}\right)$ instead of $\bigotimes_\Lambda\left(\begin{smallmatrix}1\\0\end{smallmatrix}\right)$, noting that $a_r\left(\begin{smallmatrix}0\\1\end{smallmatrix}\right) = 0$, and for every $i$, $\sigma_i^{(z)}v = -v$.

Finally, it is evident from the form of $\tilde{l}$ and $\tilde{r}$, that $\abs{r - l}\geq \frac{\abs{\tilde{r}-\tilde{l}}}{2}-1$.
\end{proof}

We are now ready to prove the following result, which implies Theorem \ref{thm:main-1}.

\begin{thm}\label{thm:v-bound}
With the underlying Hamiltonian being \eqref{eq:haml}, assume that the couplings $\mu_j \in \mathbb{R}\setminus \{0\}$, $\gamma_j \in \mathbb{R} \setminus\{\pm1\}$, $\nu_j \in \mathbb{R}$  form periodic sequences, not necessarily of the same period, let $M^{(\infty)}$ denote the limit of $M^{(\Lambda)}$ as $\Lambda \rightarrow\mathbb{Z}$, and let $Q$ be the operator associated to $M^{(\infty)}$ as in Corollary \ref{cor:lr}.

If for some $n_0\in\mathbb{N}$, the Lieb-Robinson bound holds with some positive constants $C$ and $\eta$ and some velocity $v$ as in Theorem \ref{thm:ns} on every lattice $\Lambda$ with $\abs{\Lambda}\geq n_0$, then $v\geq \norm{Q}$.
\end{thm}

\begin{proof}
Fix $C, \eta > 0$, $v\ge 0$ independently of the lattice length and the choice of observables $(A, B)\in\mathcal{D}^{(\Lambda)}$ as in Theorem \ref{thm:ns}. Since $\mu_j \neq 0$ and $\gamma_j \neq \pm 1$, by \eqref{eq:blocks}, $\det \Gamma_j \neq 0$, so Corollary \ref{cor:lr} applies to $M^{(\infty)}$. For $\epsilon> 0$, fix $\tilde{C}$, $T_0 > 0$, and $K\in\Z$ as in Corollary \ref{cor:lr}. Denote the evolution $e^{-itM^{(\infty)}}$ by $M^{(\infty)}(t)$.

For large enough $\Lambda$, by Lemma \ref{thm:propagator-lower}, we have, for an appropriate pair $(A, B)\in\mathcal{D}^{(\Lambda)}$ with $\norm{A} = \norm{B} = 1$,
\begin{align*}
\abs{M_{\tilde{l},\tilde{r}}^{(\Lambda)}(t)}\leq \mathbf{P}_t^{(\Lambda)}(A, B)\leq C e^{-\eta(\abs{r-l} + v\abs{t})},
\end{align*}
with $\set{r}=\supp(B)$, $l = \max\set{s\in\supp(A)}$ and $\abs{r-l}\geq \frac{\abs{\tilde{r}-\tilde{l}}}{2}-1$. Taking $\Lambda \to \mathbb{Z}$, this implies
\begin{equation}\label{Minfty}
\abs{M_{\tilde{l},\tilde{r}}^{(\infty)}(t)}  \leq C e^{-\eta\left(\frac{\abs{\tilde r-\tilde l}}2 -1 + v\abs{t}\right)}
\end{equation}
for all $\tilde l < \tilde r$ and all $t\in \mathbb{R}$.

By Corollary \ref{cor:lr}, for all $t > T_0>0$, there exist $\tilde{l} = \tilde{l}(t), \tilde{r} = \tilde{r}(t)\in \Z$ such that
\begin{align*}
\abs{\tilde{l}}\leq K, \hspace{2mm} 2(\norm{Q}-\epsilon)t\leq \abs{\tilde{r}}\leq 2\norm{Q}t+1,\hspace{2mm}\text{and}\hspace{2mm}\abs{M^{(\infty)}_{\tilde{l},\tilde{r}}(t)}^2\geq \frac{\tilde{C}}{t}.
\end{align*}
Then $\lvert \tilde r - \tilde l \rvert \ge 2(\lVert Q\rVert - \epsilon)t - K$, so plugging into \eqref{Minfty} we get
\[
\sqrt{\frac{\tilde C}t} \le C e^{-\eta\left((\lVert Q\rVert - \epsilon)t - K/2 -1 + v t\right)}.
\]
For this inequality to hold for arbitrarily large $t$, we must have $v \ge \lVert Q \rVert - \epsilon$; otherwise the right-hand side would decay exponentially as $t\to\infty$, which would be a contradiction. Since $\epsilon >0$ was arbitrary, this concludes the proof.
\end{proof}

\subsection{Concluding Remarks}\label{sec:remarks}

Notice that the zero velocity Lieb-Robinson bound (at least for the pairs of observables from $\mathcal{D}^{(\Lambda)}$) implies uniform dynamical localization for $M^{(\infty)}$. More precisely, with the underlying Hamiltonian being \eqref{eq:haml} and the notation from Theorem \ref{thm:v-bound}, from Lemma \ref{thm:propagator-lower} we get

\begin{thm}\label{thm:hss-converse}
Let $C, \eta > 0$ fixed and independent of $\Lambda$ and any pair of observables in $\mathcal{D}^{(\Lambda)}$. Assume that the Lieb-Robinson bound holds for all $\Lambda$ with the constants $C$ and $\eta$ and the velocity $v = 0$. Then $M^{(\infty)}(t)$ is uniformly dynamically localized. That is,
\begin{align*}
\sup_{t\in\R}\set{M_{l, r}^{(\infty)}(t)} \leq Ce^{-\eta(r-l)}
\end{align*}
for all $l < r \in \Z$.
\end{thm}

Notice that Theorem \ref{thm:hss-converse} is essentially the converse of Theorem \ref{thm:h-s-s}. Indeed, this implies Theorem \ref{thm:main-2}.

On the other hand, (the proof of) \cite[Theorem 7.5]{dRio96} shows that uniform dynamical localization implies pure point spectrum with uniformly localized eigenfunctions. Thus we have

\begin{prop}\label{thm:pp-spectrum}
Under the hypotheses of Theorem~\ref{thm:hss-converse}, $M^{(\infty)}$ has pure point spectrum with uniformly localized eigenfunctions.
\end{prop}

The conclusion of Proposition~\ref{thm:pp-spectrum} is so strong that it is only satisfied by rather special block Jacobi matrices $M^{(\infty)}$. We have already discussed the periodic case, in which we have a strictly positive lower bound for the Lieb-Robinson velocity, but the conclusion of Proposition~\ref{thm:pp-spectrum} also fails, for example, for many almost periodic cases and subshift-type models that arise in the study of quasicrystals. Thus, in all these cases we will not have a zero-velocity Lieb-Robinson bound. That is, any Lieb-Robinson bound that holds for such a model must have positive velocity.

\section{Limit-Periodic Schr\"odinger Operators}\label{Slimper}

We begin this section by proving two lemmas in preparation for the proof of Theorem~\ref{Tgeneric}. We will need a lemma about the stability of the time evolution of exponentially decaying initial states $\psi$ under small $\ell^\infty$ perturbations to the potential. This is inspired by Lemma~7.2 of Last~\cite{Last96}, which considers the case $\psi=\delta_0$, $p=2$.

\begin{lemma}\label{Lest}
Let $W$ be a bounded potential and let $t>0$, $p>0$ and $m\in \mathbb{N}$. For any $\epsilon > 0$, there is a value of $\delta>0$ such that for all potentials $V$ with $\lVert V - W \rVert_\infty < \delta$ and all $\psi \in \ell^2(\mathbb{Z})$ with
\begin{equation}\label{est9}
\lvert \psi(n) \rvert \le m e^{- \frac 1m \lvert n\rvert}, \quad \forall n\in\mathbb{Z},
\end{equation}
we have
\[
\left\lvert \langle \psi, e^{it (\Delta+V)} \lvert X \rvert^p e^{-it(\Delta+V)} \psi \rangle - \langle \psi, e^{it (\Delta+W)} \lvert X \rvert^p e^{-it(\Delta+W)} \psi \rangle  \right\rvert < \epsilon.
\]
\end{lemma}

\begin{proof} We will assume that $\delta \in (0,1)$ and $\lVert V- W \rVert_\infty < \delta$. Denote $H_1 = \Delta+W$, $H_2 = \Delta+V$. Starting with
\begin{align*}
\langle \psi, e^{it H_j} \lvert X \rvert^p e^{-itH_j} \psi \rangle & = \sum_{n\in\mathbb{Z}} \lvert n \rvert^p \lvert \langle \delta_n, e^{-itH_j} \psi \rangle \rvert^2
\end{align*}
for $j=1$ and $j=2$, subtracting and using
\[
\left \lvert \lvert \langle \delta_n, e^{-itH_1} \psi \rangle \rvert^2 - \lvert \langle \delta_n, e^{-itH_2} \psi \rangle \rvert^2 \right\rvert \le 2 \left \lvert  \langle \delta_n, e^{-itH_1} \psi \rangle - \langle \delta_n, e^{-itH_2} \psi \rangle \right\rvert,
\]
we obtain
\begin{align*}
& \left\lvert \langle \psi, e^{it H_1} \lvert X \rvert^p e^{-itH_1} \psi \rangle - \langle \psi, e^{it H_2} \lvert X \rvert^p e^{-it H_2} \psi \rangle  \right\rvert \le 2 \sum_{n\in\mathbb{Z}} \sum_{k=0}^\infty \lvert n \rvert^p \frac{t^k}{k!} \lvert \langle \delta_n, (H_1^k - H_2^k) \psi \rangle \rvert.
\end{align*}
We now wish to estimate the terms of this series. For any $k =0,1,2,\dots$, we have $\lVert H_1^k - H_2^k \rVert \le k (C_1 + \delta)^{k-1} \delta$, where $C_1 = 2 + \lVert W \rVert_\infty \ge \lVert H_1 \rVert$. This immediately implies the estimate
\begin{equation}\label{est1}
\lvert \langle \delta_n, (H_1^k - H_2^k) \psi \rangle \rvert \le k (C_1+1)^{k-1} \delta.
\end{equation}
However, for terms with $\lvert n \rvert \ge 2k$, we will use a better estimate  which follows from noting that $(H_1^k - H_2^k) \delta_n$ is supported on $[n-k, n+k] \cap \mathbb{Z}$ so
\[
\lvert \langle \delta_n, (H_1^k - H_2^k) \psi \rangle \rvert = \lvert \langle (H_1^k - H_2^k) \delta_n, \psi \rangle \rvert \le \lVert (H_1^k - H_2^k) \delta_n \rVert_2  \lVert \psi\vert_{[n-k,n+k]} \rVert_2
\]
so
\begin{equation}\label{est2}
\lvert \langle \delta_n, (H_1^k - H_2^k) \psi \rangle \rvert  \le k (C_1+1)^{k-1}\delta  C_2 e^{-\frac 1m (\lvert n\rvert -k)}
\end{equation}
where $C_2 = m \left( \sum_{l=0}^\infty e^{-2l /m} \right)^{1/2}$.

Using \eqref{est2} for $\lvert n \rvert \ge 2k$ and \eqref{est1} for $\lvert n \rvert < 2k$, we obtain
\begin{align*}
& \left\lvert \langle \psi, e^{it H_1} \lvert X \rvert^p e^{-itH_1} \psi \rangle - \langle \psi, e^{it H_2} \lvert X \rvert^p e^{-it H_2} \psi \rangle  \right\rvert  \\
& \le \sum_{k=0}^\infty \sum_{\lvert n \rvert < 2k}  \lvert n \rvert^p \frac{t^k}{k!}  k (C_1+1)^{k-1} \delta + \sum_{k=0}^\infty \sum_{\lvert n \rvert \ge 2k}  \lvert n \rvert^p \frac{t^k}{k!} k (C_1+1)^{k-1}\delta  C_2 e^{-\frac 1m (\lvert n\rvert -k)}.
\end{align*}
For the first of those sums, we use
\[
\sum_{\lvert n \rvert < 2k} \lvert n \rvert^p \le C_3 k^{p+1}
\]
and for the second sum, we use $\lvert n \rvert^p \le C_4 e^{\lvert n\rvert / (2m)}$ and
\[
\sum_{\lvert n \rvert \ge 2k} \lvert n \rvert^p e^{-\frac 1m (\lvert n\rvert -k)} \le \sum_{\lvert n \rvert \ge 2k} C_4 e^{-\frac 1{2m} (\lvert n\rvert -2k)} \le C_5
\]
for some constants $C_3, C_4, C_5$ independent of $n$ and $k$. Thus, the difference is estimated above by
\begin{align*}
& \left\lvert \langle \psi, e^{it H_1} \lvert X \rvert^p e^{-itH_1} \psi \rangle - \langle \psi, e^{it H_2} \lvert X \rvert^p e^{-it H_2} \psi \rangle  \right\rvert  \\
& \le \delta \left( \sum_{k=0}^\infty C_3 k^{p+1} \frac{t^k}{k!}  k (C_1+1)^{k-1} + \sum_{k=0}^\infty C_5 \frac{t^k}{k!} k (C_1+1)^{k-1} C_2 \right).
\end{align*}
Since both series on the right-hand side are convergent, it is now clear that for sufficiently small $\delta$, we can ensure that the right-hand side is smaller than $\epsilon$, which concludes the proof.
\end{proof}

\begin{lemma}\label{Lest2}
For any periodic potential $W$, any $p\in (0,\infty)$ and any $m\in \mathbb{N}$, there exist constants $T_{p,m,W}\ge m$ and $\delta_{p,m,W} >0$ such that the following holds: if $\psi \in \ell^2(\mathbb{Z})$ obeys $\lVert \psi \rVert_2 = 1$ and \eqref{est9}, and $V$ is any potential such that $\lVert V - W \rVert_\infty < \delta_{p,m,W}$, then
\[
\langle \psi, e^{iT_{p,m,W} (\Delta+V)} \lvert X\rvert^p e^{-iT_{p,m,W} (\Delta+V)} \psi \rangle > \frac{T_{p,m,W}^p}{\log T_{p,m,W}}.
\]
\end{lemma}

\begin{proof}
We begin by showing a preliminary estimate --- a lower bound on the size of $\psi(n)$ outside some box. Let $\psi$ be as in the statement of the lemma. This implies that for large enough $K_m \in \mathbb{N}$,
\begin{equation}\label{est4}
\sum_{\lvert n\rvert \ge K_m} \lvert \psi(n) \rvert^2 \le \sum_{\lvert n\rvert \ge K_m} m^2 e^{- \frac{2\lvert n \rvert}m} \le 2 m^2 \int_{K_m-1}^\infty  e^{- \frac{2x}m} dx \le m^3 e^{-\frac{2 (K_m-1)}m} \le \frac 1{16}.
\end{equation}
Let $j:\mathbb{R} \to \mathbb{R}$ be a bounded continuous function with $0\le j \le 1$, $\supp j \subset [-2,2]$ and $j = 1$ on $[-1,1]$. Let $j_\epsilon(x) = j(x/\epsilon)$. By Theorem~\ref{Tballistic},
\[
\lim_{t \to \infty} j_\epsilon \left(\frac{X(t)}t \right) \delta_k  = j_\epsilon(Q) \delta_k.
\]
Since $j_\epsilon \le \chi_{[-2\epsilon, 2\epsilon]}$ and $\Ker B = \{0\}$, we can pick $\epsilon >0$ such that
\[
\lVert j_\epsilon(Q) \delta_k \rVert \le \frac 1{(64 K_m)^{1/2}}
\]
for all $k \in (-K_m, K_m) \cap \mathbb{Z}$, so that for $t$ large enough,
\[
\left \lVert \chi_{(-\epsilon,\epsilon)} \left(\frac{X(t)}t \right) \delta_k \right\rVert \le \left \lVert j_\epsilon \left(\frac{X(t)}t \right) \delta_k \right\rVert \le \frac 1{(32 K_m)^{1/2}}.
\]
Then
\begin{align*}
\left\lVert \chi_{(-\epsilon,\epsilon)} \left(\frac{X(t)}t \right) \psi \right\rVert &  \le \sum_{\lvert k \rvert < K_m} \lvert \psi(k) \rvert \left\lVert \chi_{(-\epsilon,\epsilon)} \left(\frac{X(t)}t \right) \delta_k \right\rVert \\
& \qquad + \left\lVert \chi_{(-\epsilon,\epsilon)} \left(\frac{X(t)}t \right) \left( \psi - \sum_{\lvert k \rvert < K_m} \psi(k) \delta_k \right) \right\rVert \\
&  \le \frac 1{(32 K_m)^{1/2}} \sum_{\lvert k \rvert < K_m} \lvert \psi(k) \rvert  + \left\lVert \psi - \sum_{\lvert k \rvert < K_m} \psi(k) \delta_k \right\rVert \\
& \le \frac 1{(32 K_m)^{1/2}} \left( 2 K_m \sum_{\lvert k \rvert < K_m} \lvert \psi(k) \rvert^2 \right)^{1/2}  + \frac 14 \\
& \le \frac 12.
\end{align*}

Since
\[
\lvert x \rvert^p \ge \epsilon^p (1-\chi_{(-\epsilon,\epsilon)}(x)),
\]
this implies that for all large enough $t$,
\[
\frac 1{t^p} \langle \psi,  \lvert X(t)\rvert^p   \psi \rangle \ge \epsilon^p \left\langle \psi, \left (1-\chi_{(-\epsilon,\epsilon)} \left(\frac{X(t)}t \right) \right)   \psi \right\rangle \ge \epsilon^p \frac 12.
\]
In particular, for large enough $t$,
\[
 \langle \psi,  e^{it (\Delta+W)} \lvert X\rvert^p e^{-it (\Delta+W)}  \psi \rangle > 2 \frac{t^p}{\log t}.
\]
Next, by using Lemma~\ref{Lest} with $\epsilon = \frac{T_{p,m,W}^p}{\log T_{p,m,W}}$, we conclude that there exists a value $\delta_{p,m,W} > 0$ such that for any $V$ with $\lVert V - W \rVert_\infty < \delta_{p,m,W}$,
\[
 \langle \psi,  e^{iT_{p,m,W} (\Delta+V)} \lvert X\rvert^p e^{-iT_{p,m,W} (\Delta+V)}  \psi \rangle > \frac{T_{p,m,W}^p}{\log T_{p,m,W}}. \qedhere
\]
\end{proof}

\begin{proof}[Proof of Theorem~\ref{Tgeneric}]
By scaling, it suffices to prove the claim for $\lVert \psi \rVert_2 = 1$.

With the notation of Lemma~\ref{Lest2}, take
\[
U_{p,m} = \bigcup_{W} B(W,\delta_{p,m,W}),
\]
with the union taken over all periodic potentials $W$. Note that $U_{p,m}$ is open and dense (it is dense because it contains all periodic potentials). Define
\[
G = \bigcap_{k=1}^\infty \bigcap_{m=1}^\infty U_{1/k,m}.
\]
Then $G$ is a dense $G_\delta$ set. Moreover, for any $V \in G$, we have an infinite sequence of $t_{1/k,m} \ge m$ such that
\begin{equation}\label{est10}
\langle \psi, e^{it_{1/k,m}(\Delta+V)} \lvert X\rvert^{1/k} e^{-it_{1/k,m} (\Delta+V)} \psi \rangle > \frac{t_m^{1/k}}{\log t_m}
\end{equation}
whenever $\lVert \psi\rVert_2 = 1$ and $\psi$ obeys \eqref{est9}. The $t_{1/k,m}$ are simply $T_{1/k,m,W}$ with some $W$ such that $\lVert V - W \rVert < \delta_{1/k,m,W}$ (whose existence follows from the definition of $G$). Any exponentially decaying $\psi$ obeys \eqref{est9} for large enough $m$, so \eqref{est10} holds for all large enough $m$, which implies that for $V \in G$, $k\in \mathbb{N}$ and exponentially decaying $\psi$,
\[
\beta^+_\psi(1/k) \ge 1.
\]
Since $\beta^+_\psi(p)$ is a monotone increasing function of $p \in (0,\infty)$ and $\beta^+_\psi(p) \le 1$ for all $p$ (see Theorem 2.22 of Damanik--Tcheremchantsev~\cite{DamanikTcheremchantsev10}), this implies that $\beta^+_\psi(p)=1$ for all $p\in (0,\infty)$.
\end{proof}

\begin{remark}
The genericity statement of Theorem~\ref{Tgeneric} can also be made in the framework of Avila~\cite{Avila09} (see also Gan~\cite{Gan10}), where one fixes a procyclic group $\Omega$ with a minimal translation $T$, and considers for any $f \in C(\Omega,\mathbb{R})$ and any $\omega \in \Omega$ the limit-periodic potential $V$ given by
\[
V(n) = f(T^n \omega).
\]
The proof proceeds with the obvious changes, and the result becomes that for generic $f\in C(\Omega,\mathbb{R})$, we have $\beta^+_\psi(p) = 1$ for all $p>0$, all exponentially decaying $\psi$ and all $\omega \in \Omega$.
\end{remark}

Having shown that $\beta^+(p)=1$ is the generic behavior, we now turn to lack of transport. We prove the following result.

\begin{thm}
There is a dense set of limit periodic functions $V$ such that for all $p>0$ and all values of the coupling constant $\lambda>0$, for the discrete Schr\"odinger operator $\Delta + \lambda V$,
\begin{equation}\label{betaminus}
\beta^-_{\delta_0}(p) = 0.
\end{equation}
\end{thm}

\begin{proof}
The proof has two main ingredients: a construction of Avila~\cite{Avila09} (see also Damanik--Gan \cite{DamanikGan10}) of limit-periodic potentials with certain exotic properties and a criterion of Damanik--Tcheremchantsev~\cite{DamanikTcheremchantsev07} for \eqref{betaminus}. We will merely show that the potentials constructed in \cite{Avila09} obey the criterion of \cite{DamanikTcheremchantsev07}.

We denote an $n$-step transfer matrix by
\[
\Phi(n,E,w) = \begin{pmatrix} E - w_{n-1} & -1 \\ 1 & 0 \end{pmatrix} \cdots \begin{pmatrix} E - w_{0} & -1 \\ 1 & 0 \end{pmatrix},
\]
\[
L(n,E,w) =  \frac 1n \log  \lVert \Phi(n, E,w) \rVert,
\]
and the Lyapunov exponent by
\[
L(E,w) =  \lim_{n\to\infty} L(n,E,w).
\]
We will also, for a finite family $W$ of potentials, denote
\[
L(E,W) = \frac 1{\#W} \sum_{w\in W} L(E,w)
\]
and analogously for $L(n,E,W)$.

The construction of \cite{Avila09} is iterative. In the $i$-th step, one constructs a finite family $W_i$ of $p_i$-periodic potentials and constants $\delta_i \in (0,1)$ and
\[
\epsilon_{i} = \frac 1{10} \min(\epsilon_{i-1}, \delta_{i-1})
\]
such that
\begin{equation}\label{qw02}
L(E,\lambda W_i) \ge \delta_i \quad \text{if }E\in \mathbb{R},\quad  \epsilon_i \le \lvert \lambda \rvert \le \epsilon_i^{-1},
\end{equation}
and
\begin{equation}\label{qw01}
\lvert L(E,\lambda W_{i}) - L(E, \lambda W_{i-1}) \rvert < \epsilon_{i}, \quad \text{if }\lvert E \rvert, \lvert \lambda \rvert < \epsilon_{i}^{-1}.
\end{equation}
The families $W_i$ lie in balls $B_i$ whose closures are nested and radii converge to $0$, so that $W_i$ converge to a single potential $V$; that potential $V$ has the desired properties. We refer the reader to \cite{Avila09}, \cite{DamanikGan10} for details of the construction. Here we merely point out the adjustments necessary for our purpose.

The estimate
\eqref{qw01} is established by using continuity of Lyapunov exponents for periodic potentials of a fixed period $p_i$, a fact which follows by noting that for $p$-periodic potentials, $L(E,w) = L(p,E,w)$, and by using continuity of $L(p,E,w)$. In this step in the construction, we can instead use continuity of $L(q,E,w)$ for all $q\le p$ to replace \eqref{qw01} by
\begin{equation}\label{qw01b}
\lvert L(p_j, E,\lambda W_{i}) - L(p_j, E, \lambda W_{i-1}) \rvert < \epsilon_{i}, \quad \text{if }\lvert E \rvert, \lvert \lambda \rvert < \epsilon_{i}^{-1}, \quad j = 1, \dots, i.
\end{equation}
This is the only change we need --- and it doesn't affect the remainder of the construction. Moreover, since $p_j$ divides $p_{j+1}$, we have trivially $L(p_j,E,\lambda W_i) = L(p_{j+1}, E, \lambda W_{i})$ for $j > i$, so we can further restate \eqref{qw01b} as
\begin{equation}\label{qw01a}
\lvert L(p_j, E,\lambda W_{i}) - L(p_j, E, \lambda W_{i-1}) \rvert < \epsilon_{i}, \quad \text{if }\lvert E \rvert, \lvert \lambda \rvert < \epsilon_{i}^{-1}, \quad j \in \mathbb{N}.
\end{equation}
Similarly, since $W_i$ consists of $p_i$-periodic potentials, the estimate \eqref{qw02} can be rewritten as
\begin{equation}\label{qw02a}
L(p_j, E,\lambda W_i) \ge \delta_i \quad \text{if }E\in \mathbb{R},\quad \epsilon_i \le \lvert \lambda \rvert \le \epsilon_i^{-1}, \quad j\ge i.
\end{equation}
Convergence of $W_i$ to $V$ implies
\[
L(p_j, E, \lambda V) = \lim_{i\to \infty} L(p_j, E, \lambda W_i).
\]
Now assume that $\lvert E\rvert < \epsilon_k^{-1}$ and $\epsilon_k < \lvert \lambda \rvert < \epsilon_k^{-1}$. Using \eqref{qw01a} and \eqref{qw02a} we obtain for $j\ge k$ the estimate
\begin{align*}
L(p_j, E, \lambda V) & \ge L(p_j, E, \lambda W_k) - \sum_{i=k+1}^\infty \lvert L(p_j, E,\lambda W_{i}) - L(p_j, E, \lambda W_{i-1}) \rvert \\
& \ge \delta_k - \sum_{i=k+1}^\infty \epsilon_i \\
& \ge \delta_k - \sum_{i=k+1}^\infty \frac 1{10^{i-k}} \delta_k \\
& \ge \frac 89 \delta_k.
\end{align*}
By the Thouless formula,
\[
L(p_j,z,\lambda V) = \int_{\mathbb{R}} \ln\lvert z-x\rvert dk(x)
\]
where $dk$ is the density of states of a corresponding $p_j$-periodic problem corresponding to $\lambda V_1, \dots, \lambda V_{p_j}$. Since $\ln \lvert z-x\rvert \ge \ln \lvert \Re z - x \rvert$, we conclude that
\[
L(p_j, z, \lambda V) \ge L(p_j,\Re z, \lambda V) \ge \frac 89 \delta_k.
\]
To summarize, we have shown that there is a sequence $p_k \to \infty$ such that for $j\ge k$, $\lvert \Re z \rvert < \epsilon_k^{-1}$ and $\epsilon_k < \lvert \lambda \rvert < \epsilon_k^{-1}$, we have
\[
\frac 1{p_j} \log \lVert \Phi(p_j, z, \lambda V) \rVert \ge C
\]
for a constant $C = \frac 89 \delta_k >0$ independent of $z$ or $j$. Obviously, this implies for any $\lambda, K, \alpha > 0$ that, choosing $k$ large enough that $K \le \epsilon_k^{-1}$ and choosing $T_j^\alpha = p_j$, we have
\[
\int_{-K}^{K} \left( \max_{1 \le n \le T_j^\alpha} \left\lVert \Phi\left(p_j, E + \frac i{T_j}, \lambda V \right) \right\rVert^2 \right)^{-1} dE \le 2K e^{-2Cp_j}  = O(p_j^{-m})
\]
for any $m\ge 1$.  In Theorem 1 of  \cite{DamanikTcheremchantsev07} together with the remark after that theorem, this is exactly the input required in order to conclude that \eqref{betaminus} holds, which concludes the proof.
\end{proof}

\begin{remark}
Finally, we note that limit-periodic potentials can exhibit much more drastic lack of transport than that in the previous theorem. P\"oschel \cite{Poschel83} has constructed examples of limit-periodic discrete Schr\"odinger operators with a complete set of uniformly localized eigenfunctions (see also a strengthening by Damanik--Gan~\cite{DamanikGan11}). That is, those operators have a complete set $\{u_k\}\subset \ell^2(\mathbb{Z})$ of eigenfunctions and there are constants $r, C>0$ (which are independent of $k$) and $m_k \in \mathbb{Z}$ such that
\[
\lvert u_k(n) \rvert \le C e^{-r \lvert n - m_k\rvert} \quad \forall k, n\in\mathbb{Z}.
\]
By del Rio--Jitomirskaya--Last--Simon~\cite{dRio96}, this is equivalent to the property of uniform dynamical localization, that is, the existence of $D, \alpha >0$ such that
\[
\sup_{t \in \R} | \langle \delta_n,  e^{-itH} \delta_m \rangle | \le D e^{-\alpha \lvert n-m\rvert}.
\]
In turn, uniform dynamical localization implies for any exponentially decaying $\psi \in \ell^2(\Z)$,
\begin{align*}
\sup_{t \in \R} \sum_{n \in \Z} |n|^p | \langle \delta_n, e^{-itH} \psi \rangle |^2 & = \sup_{t \in \R} \sum_{n \in \Z} |n|^p | \langle e^{itH} \delta_n, \sum_{m \in \Z} \langle \delta_m , \psi \rangle \delta_m \rangle |^2 \\
& = \sup_{t \in \R} \sum_{n \in \Z} |n|^p \left\lvert \sum_{m \in \Z} \langle \delta_n,  e^{-itH} \delta_m \rangle \langle \delta_m , \psi \rangle \right\rvert ^2 \\
& \le D \sum_{n \in \Z} |n|^p \left( \sum_{m \in \Z} e^{-\alpha\lvert n -m\rvert} \lvert \langle \delta_m , \psi \rangle \rvert \right) ^2 \\
& \le D C_\alpha \|\psi\| \sum_{n \in \Z} |n|^p \left( \sum_{m \in \Z} e^{-\alpha\lvert n -m\rvert} \lvert \langle \delta_m , \psi \rangle \rvert \right)  \\
& = D C_\alpha \|\psi\| \sum_{m \in \Z} \lvert \langle \delta_m , \psi \rangle \rvert \sum_{n \in \Z} |n|^p e^{-\alpha\lvert n -m\rvert} \\
& < \infty.
\end{align*}
In particular, we have $\beta^\pm_\psi(p) = 0$ for every exponentially decaying $\psi \in \ell^2(\Z)$ and every $p > 0$.
\end{remark}

\section{Acknowledgment}

The authors would like to acknowledge helpful discussions with Robert Sims and G\"unter Stolz.

\bibliographystyle{amsplain}

\providecommand{\bysame}{\leavevmode\hbox to3em{\hrulefill}\thinspace}
\providecommand{\MR}{\relax\ifhmode\unskip\space\fi MR }
\providecommand{\MRhref}[2]{%
  \href{http://www.ams.org/mathscinet-getitem?mr=#1}{#2}
}
\providecommand{\href}[2]{#2}

\end{document}